\documentclass[10pt,reqno]{amsart}
\usepackage{amscd,amssymb,amsmath,latexsym,enumerate}
\usepackage[mathscr]{euscript}
\usepackage{epsfig}
\usepackage{fancybox}
\usepackage{verbatim}
\usepackage{tikz}
\usepackage{tikz-cd}
\usepackage{hyperref}
\usepackage{amsfonts}
\usepackage{color}
\usepackage{amssymb}
\usepackage{amsthm}
\usepackage{mathtools}

\usepackage{hyperref}
\usepackage{backref}

\title[Adiabatic approach for bulk-defect correspondences]{A space-adiabatic approach for bulk-defect correspondences in lattice models of topological insulators}

\author[D. Ojito]{Danilo Polo Ojito}
\address{Department of Physics and Department of Mathematical Sciences, Yeshiva University 
	\\New York, NY 10016, USA \\
	\href{mailto:danilo.poloojito@yu.edu}{danilo.poloojito@yu.edu}}

\author[E. Prodan]{Emil Prodan}

\address{Department of Physics and Department of Mathematical Sciences 
	\\Yeshiva University 
	\\New York, NY 10016, USA \\
	\href{mailto:prodan@yu.edu}{prodan@yu.edu}}

\author[T. Stoiber]{Tom Stoiber}

\address{Department of Physics and Department of Mathematical Sciences, Yeshiva University 
	\\New York, NY 10016, USA \\
	\href{mailto:tom.stoiber@yu.edu}{tom.stoiber@yu.edu}}

\vspace{.5cm}

\date{\today}

\newtheorem{theorem}{Theorem}[section]
\newtheorem{definition}[theorem]{Definition}
\newtheorem{proposition}[theorem]{Proposition}
\newtheorem{lemma}[theorem]{Lemma}

\newtheorem{remark}[theorem]{Remark}
\newtheorem{example}[theorem]{Example}

\newcommand{\CM}{{\mathbb C}}

\newcommand{\RM}{{\mathbb R}}
\newcommand{\SM}{{\mathbb S}}
\newcommand{\TM}{{\mathbb T}}
\newcommand{\ZM}{{\mathbb Z}}

\newcommand{\KM}{{\mathbb K}}

\newcommand{\DM}{{\mathbb D}}

\newcommand{\Aa}{{\mathcal A}}

\newcommand{\Pp}{{\mathcal P}}

\newcommand{\Bb}{{\mathcal B}}

\newcommand{\Oo}{{\mathscr O}}
\newcommand{\Tt}{{\mathcal T}}

\newcommand{\Cc}{{\mathscr C}}

\newcommand{\Ll}{{\mathcal L}}

\newcommand{\one}{{\bf 1}}

\newcommand{\Tr}{\mbox{\rm Tr}}
\newcommand{\ev}{{\mbox{\rm ev}}}

\newcommand{\Ch}{{\rm Ch}} 
 
\newcommand{\Ker}{{\rm Ker}} 
 
\newcommand{\sgn}{{\rm sgn}} 
 
\newcommand{\diag}{{\rm diag}}

\newcommand{\difd}{\textup{d}}

\providecommand{\abs}[1]{\left \lvert#1 \right \rvert} 
\providecommand{\norm}[1]{\left \lVert#1 \right \rVert}

\usetikzlibrary{calc,decorations.markings}

\begin{document}

\maketitle

\begin{abstract}
In space-adiabatic approaches one can approximate Hamiltonians that are modulated slowly in space by phase-space functions that depend on position and momentum. In this paper, we establish a rigorous relation between this approach and the operator-theoretic approach for topological insulators with defects, which employs  $C^*$-algebras and operator K-theory. Using such tools, we show that by quantizing phase-space functions one can construct lattice Hamiltonians  which are gapped at certain spatial limits and carry protected states at defects such as boundaries, hinges, and corners. Moreover, we show that the topological invariants that protect the latter can be computed in terms of the symbol functions. This enables us to compute boundary maps in K-theory that are relevant for bulk-defect correspondences. 
\end{abstract}

\tableofcontents
\section{Introduction}\label{sec:intro}

Let $\Omega$ be a locally compact separable Hausdorff space supplied with a $\ZM^d$-action $(x,\omega) \mapsto x\triangleright\omega$, and $\alpha$ the induced $\ZM^d$-action  on $C_0(\Omega)$ given by $(\alpha_{-x}(f))(\omega)=f(x\triangleright \omega)$.  We can consider the crossed product algebra $C_0(\Omega)\rtimes_\alpha \ZM^d$ as the universal $C^*$-algebra generated by formal Fourier series
$$a \;=\; \sum_{q\in \ZM^d} a_q u^q$$
with commuting unitaries $u^q=u_1^{q_1}\cdots u_d^{q_d}$ representing the generators of the $\ZM^d$-action and coefficients $a_q\in C_0(\Omega)$ satisfying the commutation relation $u^x a_q u^{-x}=\alpha_x(a_q)$. If $\KM(\mathcal{H})$ is the algebra of compact operators over a separable Hilbert space $\mathcal{H}$, then the stabilization $\KM(\mathcal{H}) \otimes (C_0(\Omega)\rtimes \ZM^d)$ of this algebra accepts a family of left-regular representations on $\mathcal{H} \otimes \ell^2(\ZM^d)$ indexed by $\Omega$,\footnote{Below, the Fourier coefficients $a_q$ belong to $\KM(\mathcal{H}) \otimes C_0(\Omega)$.}
\begin{equation}\label{Eq:PiOmega}
	\pi_\omega(a)= \sum_{q \in \ZM^d}  A_q S_q, \quad 
A_q = \sum_{x \in \ZM^d} a_q(x \triangleright \omega) \otimes |x\rangle \langle x|, \quad S_q |x\rangle = |x+q\rangle.
\end{equation}
These representations depend continuously on $\omega$ in the strong operator topology and satisfy the covariance relation $S_x \pi_\omega(a)S_x^\ast = \pi_{x\triangleright \omega}(a)$. Due to the universal property of crossed products, the image of $\pi_\omega$ is therefore itself a crossed product algebra, namely $\pi_\omega(C_0(\Omega)\rtimes_\alpha \ZM^d)= C_0(\overline{\Oo(\omega)})\rtimes_\alpha \ZM^d$ where $\Oo(\omega)=\ZM^d\triangleright\omega$ is the orbit under the $\ZM^d$-action and one takes the closure in $\Omega$. If $\Omega$ is larger than the closure of a single orbit then there is no faithful irreducible representation on $\ell^2(\ZM^d)$, however, the direct integral of $(\pi_\omega)_{\omega\in \Omega}$ w.r.t. a Radon measure with full support on $\Omega$ results in a faithful representation of the crossed product. It is important to note that a representation $\pi_\omega$ or a family of such representations enables one to localize an element of the crossed product to $\ZM^d$-invariant subsets of $\Omega$.

Representations~\eqref{Eq:PiOmega} are close in spirit to the tight-binding lattice models with a space $\mathcal{H}$ of local degrees of freedom, used in condensed matter physics. However, as it is quite apparent in Eq.~\eqref{Eq:PiOmega}, the crossed product enables one to enforce a space dependence on the coupling matrices $W_q$. It is then of no surprise that such crossed products proved useful in the analysis of disordered tight-binding models \cite{SSt, PSbook, Dani, Deni1, Dani2}. In this paper, however, we will emphasize their use in the analysis of defects, understood as space-modulations of the coupling matrices. Since these physical situations do not appear very often in the framework of crossed product algebras, a few examples are in place.
\begin{example}
	{\rm
		A translation-invariant Hamiltonian is described by the one-point space $\Omega=\{*\}$ in which case $C_0(\Omega)\rtimes \ZM^d=\CM\rtimes \ZM^d \simeq C(\TM^d)$ by Fourier transform, with the continuous functions on the torus $\TM^d=\RM^d/\ZM^d$. }
\end{example} 
\begin{example}\label{Ex:Interface1}
	{\rm 
		A flat interface between two materials produces a domain wall. If the interface is perpendicular to the first direction of real-space $\RM^d$, then a general interface Hamiltonian can always be written in the form $H=\sum_{q \in \ZM^d}  W_q S_q $, with
		\begin{equation}\label{Eq:CouplingMat}
			W_q = \sum_{x\in \ZM^d} f_q(x_1) \otimes |x \rangle \langle x|, \quad f_q: \ZM \to \KM(\mathcal{H}), \quad \lim_{s \to \pm \infty} f_q(s)= w_q^\pm.
		\end{equation}
		Here, $w_q^\pm \in M_n(\CM)$ are coupling matrices for the periodic Hamiltonians $H^\pm= \sum_{q \in \ZM^d} w_q^\pm \otimes S_q$ of the two materials. Taking $\Omega= \ZM \sqcup \{\infty\} \sqcup \{-\infty\}$ to be the two-point-compactification of $\ZM$ with the $\ZM^d$-action $x\triangleright y = y-x_1$, then $H=\pi_0(h)$ for $h =\sum_{q \in \ZM^d} f_q \otimes S_q\in \KM\otimes (C(\Omega)\rtimes \ZM^d)$. Evaluation at the infinite points recovers the two asymptotic translation-invariant Hamiltonians $H_\pm = \pi_{\pm\infty}(h)$.}
\end{example}

\begin{example}\label{Ex:PointDefect}{\rm Let $\DM^d$ be the $d$-dimensional closed ball with boundary $\partial \DM^d=\SM^{d-1}$. Let $\chi : \RM^d \to \DM^d\setminus \SM^{d-1}$ be the homeomorphism $\chi(x)=x/(1+|x|)$ with inverse $\chi^{-1}: \DM^d \setminus \SM^{d-1} \to \RM^d$,  $\chi^{-1}(\zeta)=\zeta/(1-|\zeta|)$. Moreover, let $\tau$ represent the action of $\RM^d$ on itself and define the action $x \triangleright \zeta = (\chi \circ \tau_x \circ \chi^{-1})(\zeta)$ on $\DM^d\setminus \SM^{d-1}$. Then this action extends uniquely to an $\RM^d$-action on $\DM^d$ which acts trivially on $\partial \DM^d$. We can thus define a configuration space $\Omega:= \ZM^d \cup \SM^{d-1}$ by including $\ZM^d\to \RM^d\to \DM^d\setminus \SM^{d-1}$ with the $\ZM^d$-action obtained by restriction of the above $\RM^d$-action. By examining~\eqref{Eq:PiOmega}, we see that the representations $\pi_\omega$ of some $h \in C(\Omega)\rtimes \ZM^d)$ generate for $\omega\in \ZM^d$ models with point defects which for $\omega\to \infty$ in a fixed radial direction converge to a translation-invariant Hamiltonian $\pi_{\omega_\infty}(h)$ with $\omega_\infty\in \SM^{d-1}$. See subsection~\ref{sec:ex_defect} for generalizations and further analysis.
	}
\end{example}

In this paper, we will take the following view:

\begin{definition}
	On the $\ZM^d$-lattice, a material with defects is encoded in a self-adjoint element of a cross product $C(\Omega)\rtimes \ZM^d$ and its matrix amplifications, where $\Omega$ is a compact $\ZM^d$-space. Its closed subset $\Omega_B$ of points fixed by the $\ZM^d$-action represents the bulk phases of the inhomogeneous material.
\end{definition}

The above definition certainly covers many of the standard defects, such as the point and interface ones, but it also gives one the means to construct interesting and possibly new textures on regular lattices via the left regular representations~\eqref{Eq:PiOmega}.

In the field of topological insulators, the bulk-defect correspondence problem for inhomogeneous materials described by a configuration space $\Omega$ consists of identifying and classifying self-adjoint Hamiltonians $h\in \KM \otimes (C(\Omega)\rtimes \ZM^d)$ that have a spectral gap when restricted to a distinguished closed $\ZM^d$-invariant subset $\Omega_\infty$ with $\Omega_B \subseteq \Omega_\infty\subset \Omega$, but have protected gapless states when restricted to the complement $\Omega\setminus \Omega_\infty$. The bulk-defect correspondence is naturally described by complex K-theory: If $\Omega_\infty \subset \Omega$ is a closed $\ZM^d$-invariant subset, then the restriction map induces a surjective $*$-homomorphism $R_{\Omega_\infty}\colon C(\Omega)\rtimes \ZM^d\to  C(\Omega_\infty)\rtimes \ZM^d$, leading to the exact sequence
\begin{equation}\label{eq: 1}
C_0(\Omega\setminus\Omega_\infty)\rtimes \ZM^d \hookrightarrow C(\Omega)\rtimes \ZM^d \stackrel{R_{\Omega_\infty}}{\twoheadrightarrow} C(\Omega_\infty)\rtimes \ZM^d.
\end{equation}
The stable homotopy classes of Hamiltonians that are gapped when localized at $\Omega_\infty$ are classified by the K-groups of $C(\Omega_\infty)\rtimes \ZM^d$, and \eqref{eq: 1} implies that there are associated boundary maps $$\partial\colon K_i(C(\Omega_\infty)\rtimes \ZM^d)\to K_{1-i}(C_0(\Omega\setminus\Omega_\infty)\rtimes \ZM^d),$$
mapping invariants of Hamiltonians that are gapped on $\Omega_\infty$ to obstructions to gap-opening for Hamiltonians on $\Omega$. The tasks are then to compute the supports of the boundary maps and to generate explicit representatives for the K-theoretic classes inside the supports. 

The K-theory classes are connected to Hamiltonians via functional calculus:
\begin{definition}
\label{def:ktheory} 
Let $h\in  M_N(C(\Omega)\rtimes \ZM^d)$ be  a self-adjoint matrix with entries in $C(\Omega)\rtimes \ZM^d$ and let its image $h\rvert_{\Omega_\infty} := R_{\Omega_\infty}(h)\in M_N(C(\Omega_\infty)\rtimes \ZM^d)$ have a spectral gap $\Delta$ containing $0$. We associate to $h$ the following $K$-theory classes:
\begin{enumerate}
    \item[(i)] The bulk  K-theory class  $$\Psi_0(h\rvert_{\Omega_\infty}) \;:=\; [P_{\leq 0}(h\rvert_{\Omega_\infty})]_0 \in K_0(C(\Omega_\infty)\rtimes\ZM^d),$$ the class of the spectral projection onto the spectrum below $0$.
    \item[(ii)] The boundary class $$\Theta_1(h)\; := \;[e^{\imath \pi f(h)}]_1 \in K_1(C_0(\Omega\setminus \Omega_\infty)\rtimes\ZM^d),$$ where $f\in C(\RM)$ is any function which takes the constant values $-1${\rm(}$+1${\rm)} below{\rm(}above{\rm)} $\Delta$.
\end{enumerate}
We say that $h$ is chirally symmetric if $N$ is even and $h$ anti-commutes with the chiral symmetry operator $J=\diag(\one_{N/2},-\one_{N/2})$. In that case, we can further define 
\begin{enumerate}
    \item[(iii)] The bulk class $\Psi_1(h\rvert_{\Omega_\infty})=[u_F]_1\in K_1(C(\Omega_\infty)\rtimes\ZM^d)$ where the so-called Fermi unitary is defined by the off-diagonal part of $\sgn(h\rvert_{\Omega_\infty})$ in the grading provided by $J$, i.e.
    $$\sgn(h\rvert_{\Omega_\infty})\;=\; \begin{pmatrix}
        0 & u_F^* \\ u_F & 0
    \end{pmatrix}, \quad J\;=\; \begin{pmatrix}
        \one & 0 \\ 0 & -\one
    \end{pmatrix}.$$
    \item[(iv)] The boundary class $$\Theta_0(h)\;=\; \left[\frac{1}{2}(Je^{\imath \pi f(h)}+\one_N)\right]_0 - \left[\frac{1}{2}(J+\one_N)\right] \in K_0(C_0(\Omega\setminus \Omega_\infty)\rtimes\ZM^d)$$ for $f\in C(\RM)$ is any antisymmetric function which takes the constant values $-1$ below $\Delta$ and $+1$ above $\Delta$.
\end{enumerate}
\end{definition}

\begin{proposition}[\cite{PSbook}, Sec.~4.3] Those K-theoretic classes are by construction related via the boundary maps of the exact sequence \eqref{eq: 1} 
\begin{equation}
\label{eq:bbc}
\partial(\Psi_i(h))\; = \;\Theta_{1-i}(h), \qquad i \in \ZM_2.
\end{equation}
\end{proposition}

If $\Theta_{1-i}(h)$ is non-trivial then $h$ will display gap-filling states inside the interval $\Delta$ stabilized by a K-theoretic invariant. Depending on $\Omega\setminus \Omega_\infty$ the corresponding states can correspond to various different defect modes of different codimensions in real-space. 
A codimension $n$ flat defect should preserve specific $d-n$ invariant directions and the Hamiltonians associated with the defect should be translation-invariant in these $d-n$ directions and decay in the remaining directions. Such a flat defect has a simple standard form:

\begin{definition}
\label{def:defect}
Let $\Omega$ be a $\ZM^d$-space and $\Omega_\infty\subset \Omega$ closed and invariant. We say that $\Omega$ contains a codimension $n$ flat defect localized at $\omega \in \Omega\setminus \Omega_\infty$ if the $\ZM^d$-orbit $\Oo(\omega):=\ZM^d \triangleright \omega$ is closed in $\Omega\setminus \Omega_\infty$ and $\ZM^d$-equivariantly homeomorphic to $\ZM^d / \Gamma$, where $\Gamma$ is a rank $d-n$ subgroup of $\ZM^d.$

\end{definition}

\begin{example}
{\rm In Example~\ref{Ex:Interface1} the orbit of $\omega\in \Omega\setminus \{-\infty,\infty\}$ is equivariantly homeomorphic to $\ZM$, hence $\omega$ localizes a codimension $1$ flat defect, which is the flat interface between the two materials. In Example~\ref{Ex:PointDefect} the orbit of $\omega\in \Omega\setminus \SM^{d-1}$ is homeomorphic to $\ZM^d$, hence it localizes a  codimension $d$ defect.}
\end{example}

If $\omega\in \Omega\setminus \Omega_\infty$ localizes such a flat defect, then we can classify the associated defect modes by the localized boundary invariants $\pi_\omega(\Theta_i(h))\in K_i(C_0(\Oo(\omega))\rtimes \ZM^d)$, which are equivalently defined in terms of the functional calculus of $\pi_\omega(h)$. Those can be characterized entirely in terms of numerical invariants:

\begin{proposition}
\label{prop:chern_rational}
Assume that the orbit $\Oo(\omega)$ is a closed subset in the subspace  $\Omega\setminus \Omega_\infty$. If the isotropy group $\Gamma_\omega \subset \ZM^d$ of $\omega$ is isomorphic to $\ZM^{d-n}$, then $\omega$ represents a codimension $n$ flat defect and
\begin{equation}\label{eq: takai}
 \pi_\omega(C_0(\Omega\setminus \Omega_\infty)\rtimes \ZM^d)\;\simeq\; C_0(\Oo(\omega))\rtimes \ZM^d \;\simeq\; \KM \otimes C(\TM^{d-n}).   
\end{equation}
Moreover, there exists a trace $\Tt_\omega$ on $C_0(\Oo(\omega))\rtimes \ZM^d$  such that 
elements of $K_i(\pi_\omega(C_0(\Omega\setminus \Omega_\infty)\rtimes \ZM^d))$ can then be distinguished by $\RM$-valued pairings with the Chern cocycles
$$\Ch^\omega_{v_1,\dots,v_m}(f_0,\dots,f_m)\;=\; \,
\sum_{\rho \in S_m} (-1)^\rho\, \Tt_\omega \big(f_0 \nabla_{v_{\rho(1)}} f_1 \ldots  \nabla_{v_{\rho(m)}} f_m\big)
\;,$$
with directions $v_1,\dots,v_m$ in $\RM^d$ and the densely defined $*$-derivation $\nabla_v$ is given in terms of the formal Fourier series as
\begin{equation}
\label{eq:derivation}
\nabla_v\Big(\sum_{x\in \ZM^d} f_x u^x\Big) \;\mapsto\; - 2\pi \imath \sum_{x\in \ZM^d} (v\cdot x)f_x u^x.
\end{equation}
\end{proposition}

\begin{proof}
   The first isomorphism in \eqref{eq: takai} follows from the definition of $\pi_\omega$ in \eqref{Eq:PiOmega}. Since $\ZM^d=(\ZM^d/\Gamma_\omega) \times \Gamma_\omega$ and $\Oo(\omega)$ is a countable discrete space, then one has  $\Oo(\omega)\simeq \ZM^d/\Gamma_\omega$ as $\ZM^d$-spaces. Therefore,
$$C_0(\Oo(\omega))\rtimes \ZM^d\;\simeq\;\big(C_0(\ZM^d/\Gamma_\omega)\rtimes (\ZM^d/\Gamma_\omega)\big)\rtimes \Gamma_\omega\;\simeq\; \KM(\ell^2(\ZM^d/\Gamma_\omega)) \rtimes \Gamma_\omega$$
where the trivial action by $\Gamma_\omega$ concludes \eqref{eq: takai}. The last isomorphism is a consequence of \cite[Theorem 4.23]{Wil}. This shows $K_i(C_0(\Oo(\omega))\rtimes \ZM^d)\simeq K_i(C(\TM^{d-n})$.

We define $\mathcal{T}_\omega$ to be the densely defined faithful lower semi-continuous trace supplied by the composition of counting measure on $\Oo(\omega)$ and the conditional expectation $ C_0(\Oo(\omega))\rtimes \ZM^d \to C_0(\Oo(\omega))$ (see \cite[Corollary VII.3.8]{Dav}). Then the Chern cocycles as defined above distinguish all elements of the K-group, which is a basic fact regarding the cyclic cohomology of $C(\TM^{d-n})$ (see Section~\ref{sec:chern} and Appendix~\ref{app:ktheory}). 
\end{proof}
The numerical invariants have a concrete interpretation in terms of quantized transport coefficients (as can be derived similarly as in e.g. \cite{PSbook}). The most important examples are the pairing with the $0$-cocycle $\Ch^\omega_\emptyset$ which measures the number (density) of bound states to the defect and the $1$-cocycle $\Ch^\omega_v$ for a direction $v$ is a quantized chiral conductivity for currents flowing in a direction $v$ which leaves the flat defect invariant.

In many interesting cases such standardized defects form important building blocks of more complicated spaces $\Omega$, for example in the presence of multiple different flat boundaries. One might have an increasing filtration
$\Omega=\bigcup_{m=0}^n \Omega_m$ by compact $\ZM^d$-invariant subsets such that
\begin{equation}
\label{eq:defect_decomp}\Omega_m\setminus \Omega_{m-1}\; \simeq \;  \bigsqcup_j (\ZM^d/\Gamma_{m,j})
\end{equation}
decomposes into multiple codimension $m$ defects for $m>0$ and $\Omega_0$ consisting of isolated points with trivial $\ZM^d$-action, all of which is glued together in a hierarchical way that is analogous to a CW-complex. For example, when one tries to construct a space $\Omega$ for an interface between two materials in which the boundary layer looks like the boundary of a quarter-space then it will have to contain one codimension two defect, two codimension one defects for the two asymptotic half-spaces and two invariant points corresponding to the translation-invariant bulk limits.

One can then systematically expand the K-theory and boundary maps in terms of the building blocks using spectral sequences. This is possible because the K-theoretic boundary maps factor through the closure of the orbits in $\Omega$, which are compact $\ZM^d$-invariant subsets, i.e. if $\Oo(\omega)=\ZM^d/\Gamma$ is a codimension $n$ defect one has a commutative diagram
$$
\begin{tikzcd}
	K_i(C(\Omega_\infty)\rtimes_\alpha \ZM^d) \arrow[d, "\pi_\omega"] \arrow[r,"{\partial}"] & \arrow[d, "{\pi}_\omega"]  K_{i-1}(C_0(\Omega\setminus\Omega_{\infty})\rtimes_\alpha \ZM^d)\\ 
	K_i(C(\overline{\Oo(\omega)}\cap \Omega_\infty)\rtimes_\alpha \ZM^d) \arrow[r,"{\partial_\omega}"] & K_{i-1}(C_0(\Oo(\omega))\rtimes_\alpha \ZM^d).
\end{tikzcd}.$$
It is therefore important to be able to determine the boundary maps for spaces $\overline{\Oo(\omega)}$ which are non-trivial compactifications of $\ZM^d/\Gamma$. Despite the simple standard form of $\Oo(\omega)$ it is generally difficult to compute such a boundary map $\partial_\omega$ using topological methods alone. Indeed, there is no general theory to perform those computations and, at best, one can attempt to find a basis of the K-groups represented by Hamiltonians gapped at $\Omega_\infty$ and compute their defect invariants. This is a difficult task to carry out analytically, especially in higher dimensions and for general geometric defects. After all, one has to rigorously prove that specific Hamiltonians have spectral gaps when restricted to $\Omega_\infty$ and compute their defect invariants both of which is generally only feasible for models that can be diagonalized exactly or have a very particular algebraic form.

In this paper, we explore a space-adiabatic Ansatz to help with those computations. The idea is that if a Hamiltonian $h=\sum_{q\in\ZM^d}a_qu^q \in M_N(C(\Omega)\rtimes_\alpha \ZM^d))$ is modulated slowly in space then the coefficient functions $a_q$ approximately commute with the unitaries $u^q$ and therefore the topological content should hopefully be roughly the same as that of a continuous function in $M_N(C(\Omega\times \TM^d))$, its {\it adiabatic symbol}. In particular it should sometimes be possible to compute the defect invariants using the boundary map $K_i(C(\Omega_\infty\times \TM^d))\to K_{1-i}((C\Omega\setminus \Omega_\infty)\times\TM^d)$ between topological spaces which is usually much more amenable to computation by algebraic-topologial methods such as spectral sequences. This is similar to the semiclassical limit for the corresponding situation in continuous space where Hamiltonians are differential operators whose spectral and topological properties become for $\hbar\to 0$ more and more determined by their symbol, a phase-space function on the cotangent bundle. More closely related is the space-adiabatic approach \cite{PST03a, PST03b} which works similarly, except that the semi-classical parameter $\hbar$ is replaced by the inverse of a characteristic length-scale such that one considers instead the limit of infinitely slow modulation in real-space. In principle space-adiabatic methods can also be used for lattice operators, however, a problem for topological methods is that the lattice itself cannot be rescaled continuously, hence one cannot canonically define an adiabatic limit for general configurations spaces $\Omega$. And indeed, the commutative algebras $C(\Omega)\rtimes_\alpha \ZM^d$ and $C(\Omega)\otimes C(\TM^d)$ in general will behave very differently from each other in K-theory, even in the simple cases of Example~\ref{Ex:Interface1} and Example~\ref{Ex:PointDefect}. 

The solution that we propose for this problem is to make the adiabatic Ansatz using slightly different models for the configuration spaces $\Omega$ than commonly used, but for which it is then possible to establish a good relation between the crossed product algebra and the adiabatic phase-space $\Omega\times \TM^d$. Let us therefore assume that we have a configuration space $\Omega$ whose $\ZM^d$-action is the restriction of an $\RM^d$-action and which moreover comes equipped with a compatible action of the multiplicative group $\RM_+$ implementing a scale  transformation $(t,\omega) \mapsto t \cdot \omega$. This allows us to consider rescaled crossed products $C(\Omega)\rtimes_{\alpha^{(t)}}\ZM^d$ for the $\ZM^d$-action $\alpha^{(t)}_x = \alpha_{tx}$, which are for $t>0$ isomorphic to each other. Inspired by \cite{ENN93,ENN96}, we then obtain a continuous field of $C^*$-algebras
\begin{equation}\label{Eq:ContField}
	\left(C(\Omega)\rtimes_{\alpha^{(t)}}\ZM^d\right)_{t\in [0,1]}
\end{equation}
with trivial action at the end-point $t=0$, where the crossed product therefore becomes the commutative algebra of phase-space functions $C(\Omega \times \TM^d)$. A quantization of a phase-space function $f\in M_N(C(\Omega \times \TM^d))$ shall be a continuous section of the field ~\eqref{Eq:ContField} whose evaluation at $t=0$ reproduces $f$. The machinery of continuous fields ensures the existence of a large number of quantizations, even without having to specify a concrete quantization map. The existence of a scaling action on $\Omega$ will be crucial to ensure that one can construct moreover quantizations with additional properties, such as spectral gaps, see Section~\ref{sec-adiabatic}.  For the typical continuous section $(h^{(t)})_{t\in [0,1]}$ consisting of self-adjoint Hamiltonians one will find that their representations $\pi_\omega(h^{(t)})$ on $\ell^2(\ZM^d)$ are modulated more and more slowly in real-space as $t\to 0$. Thus the continuous field allows us to rigorously define the space-adiabatic limit on an operator-algebraic level. Moreover, quantization does induce a natural homomorphism on the level of K-theory $Q:K_i(C(\Omega\times \TM^d))\to K_i(C(\Omega)\rtimes_\alpha \ZM^d)$ which behaves naturally with respect to boundary maps.

To apply this formalism to the computation of boundary maps of a $\ZM^d$-space $\Omega$ as in Example~\ref{Ex:Interface1} or \ref{Ex:PointDefect} one has to embed it into a larger space which has an $\RM^d$- and an $\RM_+$-action. This can be done by embedding each codimension $n$ defect $\Oo(\omega)\simeq \ZM^d/\Gamma$ into a copy of $\RM^n$ which carries an affine action of $\RM^d$. For a configuration space which decomposes like \eqref{eq:defect_decomp} the natural model for an adiabatic configuration space is therefore a CW-complex with actions of $\RM^d$ and $\RM_+$ which act affinely on the cells. Moreover, the affine actions should leave certain lattice directions invariant, which imposes a rationality condition. A CW-complex like this will be called flat and rational in this work, see Section~\ref{sec:flat}.

We can now announce our main result:
\begin{theorem}\label{Th:Main}
\label{th:main}
Let the compact configuration space $\Omega$ be a CW-complex which is flat and rational at level $n$, in which case the $n$-skeleton decomposes as $\Omega_n\setminus \Omega_{n-1}\simeq \sqcup_j \RM^n_j$ with some labeled copies of $\RM^n$. If $H: \Omega \times \TM^d \mapsto M_N(\CM)$ is a self-adjoint function whose restriction to $\Omega_{n-1}\times \TM^d$ has a spectral gap, then there exists an associated adiabatic quantization, $(H^{(t)})_{t\in [0,1]}\in \left(M_N(C(\Omega)\rtimes_{\alpha^{(t)}}\ZM^d)\right)_{t\in [0,1]}$, i.e. a continuous self-adjoint section, with the following properties:
\begin{enumerate}
\item[\noindent (i)] The evaluation $H^{(1)} \in M_N(C(\Omega_{n-1})\rtimes \ZM^d)$ has a spectral gap.
\item[(ii)] If $H$ is invariant under a finite linear group action then each $H^{(t)}$ can also be chosen invariant.
\item[(iii)] Let $\omega\in \RM^n_j$ be an element of one of the $n$-cells. Then the $\ZM^d$-orbit $\Oo(\omega)$ in $\Omega_n \setminus \Omega_{n-1}$ is closed and discrete, with the orthogonal complement of the isotropy group $\Gamma_\omega$ generated by unit vectors $\lambda_1,\dots,\lambda_{n} \in \RM^d$, hence it represents a codimension $n$ defect.

\noindent \item[(iv)] For $\omega\in \RM^n_j$, all pairings of $\Theta_{1-i}(\pi_\omega(H^{(1)}))$ with the complete set of Chern cocycles of Proposition~\ref{prop:chern_rational} can be computed from the adiabatic symbol as
\begin{equation}
	\label{eq:main_eq}
	\begin{aligned}
	& \langle \pi_\omega(\Theta_{1-i}(H^{(1)})), \Ch^\omega_{v_1,...,v_m}\rangle \; \\
	& \qquad \quad = \;\tfrac{1}{\mathrm{Vol}(\RM_j^n/\ZM^d)}\langle \Theta_{1-i}(H))\rvert_{\RM_j^n}, \Ch_{v_1,\dots,v_m,\lambda_1,\dots,\lambda_n}\# \Ch_{\RM^n_j}\rangle
\end{aligned}
\end{equation}
where one takes the cup product with the canonical Chern cocycle on the $n$-cell $\RM^n_j$ and normalizes by dividing by the Lebesgue volume of a fundamental domain for the $\ZM^d$-action. Here one makes the identification  $$ K_{1-i}(C_0(\Omega_n\setminus \Omega_{n-1})\otimes C(\TM^d))\;\simeq\; \bigoplus_j K_{1-i}(C(\TM^d)\otimes C_0(\RM^n_j))$$
and restricts $\Theta_{1-i}(H)\in K_{1-i}(C_0(\Omega_n\setminus \Omega_{n-1})\otimes C(\TM^d))$ to one of the cells.
\end{enumerate}

\end{theorem}

 Theorem~\ref{Th:Main} provides us with a direct method to construct lattice Hamiltonians that have known asymptotic limits, spectral gaps and Chern numbers, which is invaluable information for K-theoretic computations. The fundamental domain appears in~\eqref{eq:main_eq} since the right-hand side of \eqref{eq:main_eq} turns out to be an average of the left-hand side over all $\omega \in \RM^n_j$ in the same $n$-cell.  Given a concrete problem with a (preferred) $\ZM^d$-space $\Omega$ one can make an adiabatic Ansatz for the construction of Hamiltonians whenever $\Omega$ (or at least some smaller closed orbit $\overline{\Oo(\omega)}$) embeds into a flat CW-complex, which need not be unique. After all, the defect Chern numbers of any lattice Hamiltonian resulting from an adiabatic construction, i.e. the left-hand side of \eqref{eq:main_eq}, engage only a single $\ZM^d$-orbit.

The CW-structure is also convenient in that it allows us to efficiently construct interesting Hamiltonians in terms of gapped adiabatic symbols on the asymptotic part $\Omega_\infty \times \TM^d$ of the phase-space, since their defect invariants \eqref{eq:main_eq} are determined by a boundary map $K_i(C(\Omega_\infty\times \TM^d))\to K_{1-i}((\Omega\setminus \Omega_\infty)\times\TM^d)$ which is comparatively easy to compute by reducing to the K-theoretic boundary map between the boundary and the interior of a disk, see Section~\ref{sec:boundary_maps}.

In section~\ref{Sec:App}, we illustrate the method for different geometries. After demonstrating that the usual bulk-interface correspondence can be derived easily through the adiabatic method, we show that one can also obtain information about higher-order boundary maps describing bulk-hinge or bulk-corner correspondence.

Our approach bears strong similarities to the use of effective models in the form of matrix-valued phase-space functions for the study of topological insulators and their gapless defect states which goes back to 
\cite{TeoKane2010}. That approach is appealing for physical applications because it avoids the intricacies of the operator-algebraic approach and, with some limitations, it can also be used for crystalline topological phases \cite{ShiozakiPRB2014,ShiozakiPRB2016}. Our work  supplies some rigorous foundations and justifications to those ideas since it shows that those effective models can often be systematically translated back into actual lattice models while behaving consistently on the level of K-theoretical topological invariants.

While our paper is focused entirely on lattice models, let us nevertheless briefly discuss the relation to pseudodifferential methods for topological insulators on continuous space (see in particular \cite{Bal18,Bal19, Bal21, Drouout21, BBD, Bal23}). For example, there are already cohomological formulas for interface and related topological invariants in terms of the pseudodifferential symbols of domain wall Hamiltonians derived by using pseudodifferential methods and index theorems. In principle, the K-theoretic results of this paper all have analogues for crossed products by $\RM^d$ instead of $\ZM^d$. Those are then the appropriate algebras for the resolvents of elliptic differential operators in the presence of domain walls or defects. The mentioned index formulas can then be understood systematically through the pairing of K-theory with cyclic cohomology and the quantization map of a continuous field of $C^*$-algebras. This deserves some further study, especially regarding the technical problems that unbounded Hamiltonians introduce into the $C^*$-algebraic formalism.

Pseudodifferential calculus and quantization arguments are also used copiously in index theory to the extent that it does not make sense even to attempt to provide a full bibliography here. Let us, however, mention Connes' tangent groupoid \cite{Connes94, Higson2008} and other deformation groupoids (e.g. \cite{Monthubert,vEY2019,DS2014}), which result in similar continuous fields of $C^*$-algebras and maps in K-theory as we consider here. For a combination of K-theory of continuous fields and cyclic cohomology in a similar spirit as employed in this work let us also highlight \cite{ENN96, KS20042,MSS2006}.

\medskip

\noindent
{\bf Acknowledgements:} The authors would like to cordially thank C. Bourne for stimulating discussions. This work was supported by the U.S. National Science Foundation through the grants DMR-1823800 and CMMI-2131760, and by U.S. Army Research Office through contract W911NF-23-1-0127, and the German Research Foundation (DFG) Project-ID 521291358.

\section{Continuous fields for adiabatic quantization}
\label{sec-adiabatic}

As described in the introduction, our main result involves configuration spaces $\Omega$ which have an action of $\underline{\rm Sim}(d):=\RM^d\rtimes \RM_+$, a subgroup of the group ${\rm Sim}(d)$ of similarity transformations on $\RM^d$. In the rest of this work all spaces $\Omega$ will be assumed to have that extra structure, whereas the discrete variants such as the $\ZM^d$-spaces from Example~\ref{Ex:Interface1}, Example~\ref{Ex:PointDefect} only appear as the restrictions to $\ZM^d$-orbits $\overline{\Oo(\omega)}$ of points $\omega\in \Omega$ anymore.

\begin{definition}
	A configuration space $\Omega$ shall be a compact Hausdorff space gifted with a continuous $\underline{\rm Sim}(d)$-action. The actions by translations $\RM^d$  are written as $(\omega, x)\mapsto x\triangleright \omega$ and the scaling $\mathbb{R}_{+}$ is written as $(t,\omega)\mapsto t\cdot \omega$. Obviously, 
$$t\cdot (x\triangleright\omega) \;=\; (tx)\triangleright (t\cdot \omega), \quad (t,x)\in \RM_+ \times \RM^d.$$ 
This results in a strongly continuous $\RM^d$-action $\alpha$ and $\RM_+$-action $\gamma$ on $C(\Omega)$ and $\gamma$ of $\RM_{+}$ such that
	$$\gamma_t \circ \alpha_x \;=\; \alpha_{t x}\circ \gamma_t.$$
\end{definition}

We recall that our interest is in the discrete crossed product $C(\Omega) \rtimes \ZM^d$, where the action of $\ZM^d \subset \RM^d$ on $C(\Omega)$ comes from the restriction of $\alpha$. This algebra encodes lattice models that depend parametrically on $\omega \in \Omega$. Even though we work with lattice models, it is fruitful to engage a hybrid approach between the usual descriptions of continuous and lattice Hamiltonians, where we have an $\RM^d$ action on the configuration space and a $\ZM^d$ action on the physical space.

\begin{definition}
\label{def:adiabatic_algebras}
Let $\Omega$ be a compact configuration space.
Consider the rescaled translation action $$\alpha^{(t)}_x(f)\;=\;\alpha_{t x}(f)$$ on $C(\Omega)$ with component-wise scaling by $t\in T:=[0,1]$. The adiabatic quantization of the symbol space $C(\Omega)\otimes C(\TM^d)$ is represented by the continuous field of $C^*$-algebras
$$\Aa\; := \;(\Aa_t)_{t\in T}\; = \;\left(C(\Omega)\rtimes_{\alpha^{(t)}} \ZM^d\right)_{t\in T}.$$
By $\Aa$ we also denote the total algebra of the field and ${\rm ev}_t\colon \Aa\to \Aa_t$ stands for the evaluation at $t\in T$. The fiber at $t=0$ is isomorphic to $C(\Omega)\otimes C(\TM^d)$.
\end{definition}

Recall that a continuous field of $C^*$-algebras $(\Aa_t)_{t\in T}$ is itself a $C^*$-algebra obtained by completing a set of distinguished sections of $\sqcup_{t\in T}\Aa_t$ with certain properties \cite[Chapter 10]{Dixmier1982} in the natural $C^*$-norm. In the present case, any family of crossed products with continuously varied action by an amenable group has a canonical structure as a continuous field for which the total algebra of the field satisfies (e.g. \cite[Theorem 2.5]{ENN93})
\begin{equation}
\label{eq:field_crossed_product} \left(C(\Omega)\rtimes_{\alpha^{(t)}} \ZM^d\right)_{t\in T}\; =\; \left(C(T)\otimes C(\Omega)\right)\rtimes_{\alpha^{(T)}} \ZM^d\end{equation}
with the action $\alpha^{(T)}_x(f)(t) = \alpha^{(t)}_x(f(t))$.
Hence a dense subset of continuous sections is precisely given by the formal Fourier series
$$t \in T \;\mapsto\; \sum_{q\in \ZM^d} f_q(t) u_t^q, \qquad f_q \in C(T, C(\Omega))$$
whose finitely many non-vanishing coefficient functions depend continuously on a parameter and with unitary generators satisfying $u_t^q f u_t^{-q}=\alpha^{(t)}_q(f)$. In the present case, the field is trivial away from $0$ in the sense that all fibers for $t>0$ are isomorphic to each other due to the scaling action of $\gamma$, moreover:
\begin{lemma}
\label{lemma:rescaling}
For any $0 < t_0 < 1$ the field 
$$\left(C(\Omega)\rtimes_{\alpha^{(t)}} \ZM^d\right)_{t\in [0,t_0]}$$
is isomorphic to $\Aa$ as a continuous field.
\end{lemma}
\begin{proof}
For any $s,t\neq 0$ the isomorphism $\rho_{s,t}:C(\Omega)\rtimes_{\alpha^{(s)}}\ZM^d \to C(\Omega)\rtimes_{\alpha^{(t)}}\ZM^d$ is given by
$$\rho_{s,t}\Big(\sum_{q\in \ZM^d} a_q u_s^q\Big) \; = \;\sum_{q\in \ZM^d} \gamma_{t/s}(a_q) u_t^q.$$
Choosing any continuous increasing function $g: [0,t_0]\to [0,1]$ which is equal to the identity on $[0,\frac{1}{2}t_0]$ the isomorphisms $\rho_{s,g(s)}$ together define an isomorphism of the continuous fields.
\end{proof}

\medskip


\begin{definition}
Consider a continuous field as in Definition~\ref{def:adiabatic_algebras}. We say that a continuous section $\hat{\xi}=(\xi_t)_{t\in T}$ of the field with $\ev_0(\hat{\xi})=\xi_0=\xi$ is an adiabatic quantization of $\xi \in \Aa_0$. For the particular case when $\xi$ is a projection/unitary/self-adjoint an adiabatic quantization shall be a continuous section which consists of projections/unitaries/self-adjoints. 
\end{definition}
For any continuous field there is by definition at least a norm-dense subalgebra of elements of $\Aa_0$ which have an adiabatic quantization and for self-adjoint elements the section can be made self-adjoint. To construct sections with additional properties we recall that sections satisfy spectral continuity, i.e. given any self-adjoint continuous section $(\xi_t)_{t\in T}$ the spectra $\sigma(\xi_t)\subset \RM$ will depend on $t$ continuously in the sense that if $\sigma(\xi_t)$ is contained in an open set for some $t$, then this will also be true for an open neighborhood of $t$ \cite[Proposition 10.3.6]{Dixmier1982}. In particular, if $\xi_t$ is invertible then this will also hold for a small neighborhood of $t$. It is well-known that a selfadjoint element of a $C^*$-algebra with a spectral gap can be deformed to a projection via functional calculus and also that continuous functional calculus maps continuous sections to continuous section \cite[Proposition 10.3.3]{Dixmier1982}. Likewise, any invertible which commutes with its adjoint up to a small enough error $\epsilon$ is close to an actual unitary by polar decomposition. Therefore one can replace any adiabatic quantization of a projection/unitary $\xi_0$ with another adiabatic quantization such that $\xi_t$ is a projection/unitary at least for small enough $t$. This is enough since triviality away from $0$ implies that one can stretch any such section to cover all of $[0,1]$, see Lemma~\ref{lemma:rescaling}. Hence any projection/unitary has an adiabatic quantization consisting of projections/unitaries. With this argument, one can show that, for any continuous field over $[0,1]$ that is trivial away from $0$, there exists a canonical homomorphism \cite{ENN93}
\begin{equation}\label{eq: Q}
Q_i\colon [\xi_0]_i \in K_i(\Aa_0) \;\mapsto\; [\xi_1]_i \in K_i(\Aa_1)
\end{equation}
which does not depend on the continuous section used. In general, it need not be injective or surjective.

As motivated in the introduction, we want to quantize functions on the phase space $\Omega \times \TM^d$ to obtain lattice Hamiltonians that have spectral gaps at specified asymptotic limits and computable Chern numbers.

\begin{definition}
Let $\Omega$ be a configuration space. An adiabatic symbol is a self-adjoint function $H$ in $C(\Omega\times \TM^d) \otimes M_N(\CM)$. For a compact invariant subset  $\Omega_\infty$ of $\Omega$ let $R^{(0)}_{\Omega_\infty}\colon C(\Omega\times \TM^d) \to C(\Omega_\infty \times \TM^d)$ be the restriction map. We say that $H$ is gapped on $\Omega_\infty\subset \Omega$ if $H|_{\Omega_\infty}:=R^{(0)}_{\Omega_\infty}(H)$ has a spectral gap around $0$, i.e. there is a compact interval $\Delta$ containing $0$ so that $\sigma\big(H|_{\Omega_\infty}\big)\cap \Delta =\emptyset$.
\end{definition}


Here and in the remainder of this work, we say that a subset of $\Omega$ is invariant if it is mapped to itself under both the $\RM^d$- and $\RM_+$-actions.
\begin{proposition}
\label{prop:quantization}
Given an adiabatic symbol $H$ and any $\epsilon>0$, there exists a continuous section  $(H^{(t)})_{t\in T}$ of the field $(M_N(C(\Omega)\rtimes_{\alpha^{(t)}}\ZM^d))_{t\in T}$ consisting of self-adjoint operators such that $\norm{H^{(0)}-H}<\epsilon$. Let $\Omega_\infty\subset \Omega$ be closed and invariant. If $H$ is gapped on $\Omega_\infty$, then the section can be chosen such that: 
\begin{enumerate}[i.]
    \item[(i)] $H^{(t)}$ is gapped in $M_N(C(\Omega_\infty)\rtimes_{\alpha^{(t)}} \ZM^d)$ for all $t$;
    \item[(ii)] Let $\Omega$ carry a linear $G$-action for some finite subgroup $G$ of $O(d)$ that leaves $\Omega_\infty$ invariant and induces an action $\eta: C(\Omega)\times G \to C(\Omega)$, then the quantization can additionally be chosen to also be $G$-invariant. Here, linear means that $\eta_g \circ \alpha_{x} = \alpha_{gx} \circ \eta_g$ on $C(\Omega)$ and $\eta_g \circ \gamma_t = \gamma_t \circ \eta_g$, which implies that the action extends to the crossed product.
\end{enumerate}

\end{proposition}
\begin{proof}
As discussed above, there is always a self-adjoint section whose endpoint $H^{(0)}$ is arbitrarily close to $H$.
By $G$-invariant we mean here that $\eta_g(H^{(t)})=U_g H^{(t)}U_g^*$ for a finite-dimensional representation $U:G\to M_N(\CM)$. By averaging over $G$ one can produce a $G$-invariant self-adjoint section $(H^{(t)})_{t\in T}$ such that $H^{(0)}$ is arbitrarily close to $H$. \\  
  The restriction map fits into the commutative diagram with exact rows
\begin{equation*}
	\begin{tikzcd}[column sep=small]
	\left(C_0(\Omega\setminus \Omega_\infty)\rtimes_{t} \ZM^d\right)_{t\in T} \arrow[d]\arrow[r,hook,""] & \left(C(\Omega)\rtimes_{\alpha^{(t)}} \ZM^d\right)_{t\in T} \arrow[d]\arrow[r,two heads,"R_{\Omega_\infty}^{(t)}"] & \left(C(\Omega_\infty)\rtimes_{\alpha^{(t)}} \ZM^d\right)_{t\in T} \arrow[d]\\
	C_0(\Omega\setminus \Omega_\infty)\otimes C(\TM^d)
	\arrow[r,hook,""]  & C(\Omega)\otimes C(\TM^d) \arrow[r,two heads,"R^{(0)}_{\Omega_\infty}"] & C(\Omega_\infty)\otimes C(\TM^d)
	\end{tikzcd}
\end{equation*} 
which implies that the restriction $R_{\Omega_\infty}^{(t)}(H^{(t)})\in M_N(C(\Omega_\infty)\rtimes_{t}\ZM^d)$ for any adiabatic symbol is again a self-adjoint $G$-invariant continuous section.  Since $R^{(0)}_{\Omega_\infty}(H^{(0)})$ has a spectral gap around $0$ and one has continuity of spectra along a continuous section one knows that  $R^{(t)}_{\Omega_\infty}(H^{(t)})$ also has a spectral gap around $0$ for small enough $t$. One can then rescale $t$ as in Lemma~\ref{lemma:rescaling} to obtain a section which is gapped at $\Omega_\infty$ for all $t\in [0,1]$.  
\end{proof} 
\begin{remark}{\rm
Proposition \ref{prop:quantization} is also true for  a $\phi$-twisted invariant adiabatic symbol $H$, i.e. $\eta_g(H^{(t)})=(-1)^{\phi(g)}U_g H^{(t)}U_g^*$ where $\phi\colon G \to \ZM_2$  is a homomorphism. This is in particular relevant for a chiral symmetric Hamiltonians which needs to satisfy an anti-symmetry $H=-JHJ$. }   \hfill    $\blacktriangleleft$  
\end{remark}


For a general configuration space $\Omega$ we just saw that it is not difficult to prove the existence of adiabatic quantizations with nice properties. Let us make the construction of such a quantization more explicit. For simplicity, let $H: \Omega \times \TM^d \to M_N(\CM)$ have uniformly finite hopping range in the sense that there is a Fourier expansion
$$H(\omega, k) \;= \;\sum_{q\in \ZM^d} f(\omega) e^{\imath k \cdot q}$$
with only finitely many non-vanishing coefficient functions $f_q\in M_N(C(\Omega))$. An adiabatic quantization is generated by the family of finite Fourier series 
$$H^{(t)} \;=\; \sum_{q\in \ZM^d} f^{(t)}_q u_t^q \in M_N(C(\Omega)\rtimes_{\alpha^{(t)}}\ZM^d)$$
with coefficients given by finitely many non-vanishing continuous functions $t\in [0,1]\mapsto f^{(t)}_q \in C(\Omega)$, $f^{(0)}_q=f_q$, and $(u^q_t)_{q\in \ZM^d}$ the unitary generators of the crossed product $C(\Omega)\rtimes_{\alpha^{(t)}}\ZM^d$, hence satisfying the commutation relation $u_t^yf^{(t)}_q u_t^{-y}  = \alpha_y^{(t)}(f^{(t)}_q)$. After symmetrizing one can read off functions $f_q^{(t)}$ which define a self-adjoint continuous section from
$$H^{(t)} \;=\; \tfrac{1}{2}\Big(\sum_{q\in \ZM^d} f_q u_t^q + u_t^{q} f^*_{-q}\Big)=\sum_{q\in \ZM^d} \tfrac{1}{2}\left(f_q + \alpha^{(t)}_{q} (f^*_{-q})\right) u_t^q \;.$$ 
In any  representation on $\ell^2(\ZM^d)$ the associated matrix elements  take the form
\begin{align*}
\langle \delta_x, \pi^{(t)}_\omega(H^{(t)}) \delta_y\rangle \;&=\; f^{(t)}_{x-y}\big(tx\triangleright\omega\big)\\
\;&=\; \tfrac{1}{2}\left(f_{x-y}\big(tx\triangleright\omega)+f^*_{y-x}\big(ty\triangleright\omega)\right), \qquad \forall x,y\in \ZM^d.
\end{align*}
Thus, they depend on the real-space location, but any characteristic length scale becomes proportional to $t^{-1}$. Hence it is justified to identify the multiplication operator $H^{(0)}\in M_N(C(\Omega) \otimes C(\TM^d))$ as the adiabatic limit of $H^{(1)} \in C(\Omega) \rtimes_\alpha \ZM^d$. To make sure that $H^{(1)}$ varies slowly enough in space for the additional properties discussed in Proposition~\ref{prop:quantization}, one may have to rescale the coefficient functions as in Lemma~\ref{lemma:rescaling} such that $t=1$ represents a large enough scale. 

\section{Flat configuration spaces}
\label{sec:flat}

Recall that a CW-complex structure on a topological space $\Omega$ is given by an increasing filtration $\Omega=\bigcup_{-1\leq d} \Omega_n$ by closed sets, $\Omega_{-1}=\emptyset$, such that every $\Omega_n$ is obtained as an iterated pushout of cells
\begin{equation*}
	\begin{tikzcd}
		\coprod_i \SM^{n-1}\arrow{r} \arrow{d} & \Omega_{n-1} \arrow[d] \\
		\coprod_i \DM^{n}  \arrow{r} & \Omega_{n}
	\end{tikzcd}
\end{equation*}

In this paper, it will be most convenient to identify the interior of the $n$-dimensional closed ball with $\RM^n$. We also give $\RM^n$ a canonical vector space structure. An action on $\RM^n$ can then be called affine if it acts via affine transformations.

\begin{definition}
\label{def:flat_cw}
	Let a configuration space $\Omega$  be a CW-complex with skeleton filtration $(\Omega_n)_{0\leq n\leq d}$, hence there is a homeomorphism
	$$\Omega_n\setminus \Omega_{n-1}\; \stackrel{\rho_{n}}{\longrightarrow}\; \bigsqcup_{i=1}^{N_n} \RM_i^n$$
	with labeled copies of $\RM^n$. We say that $\Omega$ is flat at level $n$ if $\rho_n$ can be chosen equivariant w.r.t. the translation and scaling action for affine actions on the copies $\RM^n_i$, i.e. for every $i$ there is a $n\times d$ matrix $\Lambda$ such that
	\begin{align*}
		x\triangleright \rho_n^{-1}(y) \;=\; \rho_n^{-1}(y + \Lambda x),\quad
		t\cdot \rho_n^{-1}(y) \;=\; \rho_n^{-1}(t y), \quad (t,x)\in \RM_+ \times \RM^d,
	\end{align*}
for $y\in \RM^n_i$. We also impose that $\Lambda$ has maximal rank. We will say, moreover, that the action is rational if the kernel of $\Lambda$ is spanned by $d-n$ linearly independent vectors in $\ZM^d$ or, equivalently, $\Lambda$ is proportional to a matrix in $M_{d\times n}(\mathbb{Q})$.
\end{definition}\label{Def:FlatCW}
If $\Omega$ is flat at every level, we will simply call it a flat configuration space. The rationality condition is important in view of Definition~\ref{def:defect}, since it implies that the $\ZM^d$-orbit of any $\omega \in \Omega_n\setminus \Omega_{n-1}$ is discrete with finite co-volume in $\Omega_n\setminus \Omega_{n-1}$ and, therefore, it represents a codimension $n$ defect for the pair $\Omega=\Omega_n$ and $\Omega_\infty=\Omega_{n-1}$. Note that the rationality condition is un-affected by scaling.

As a simple example the configuration space from Example~\ref{Ex:PointDefect} includes into the disk $\DM^d=\RM^d\sqcup \SM^{d-1}$ with the affine action of $\RM^d$ on itself. This obviously has a CW-structure, which is flat at the top level. That will also be true for the generalizations in section~\ref{sec:ex_defect}.

In the rest of the section, we show that the structure introduced in Definition~\ref{Def:FlatCW} appears quite naturally when considering geometric boundaries defined by polyhedral cones. 

Let us first consider the compact metric space $\mathscr{C}(\RM^d)$  of closed subsets of $\RM^d$ endowed with Fell topology \cite{FellPAMS1962}, i.e. the topology induced by the Hausdorff metric on $\mathscr{C}(\SM^d)$ upon including $\mathscr{C}(\RM^d)\to \mathscr{C}(\SM^d)$ via stereographic projection \cite{FHK, LS}. If $\mathcal{L}\in \mathscr{C}(\RM^d)$ we consider its orbit space $ \mathscr{O}_{\RM^d}(\mathcal{L}):=\{\mathcal{L}-x\;|\;x\in  \RM^d\}$. If $\Ll$ is scale-invariant then a configuration space is defined by the compact space
\begin{equation}
    \Omega_\mathcal{L}\;:=\;\overline{\mathscr{O}_{\RM^d}(\mathcal{L})}
\end{equation}
where the closure is taken with respect to the Fell topology and the translation and scaling action are inherited from $\RM^d$.





For each $v\in \SM^d$, we define the half-plane $\mathcal{L}_v := \{x \in \RM^d \;|\; v\cdot x\geq 0\}$. A convex cone can be defined as a non-empty intersection of half-spaces $\mathcal{L}=\bigcap_{i=1}^m \mathcal{L}_{\lambda_i}$. If $m=d$ and the normal vectors $\lambda_i$ are linearly independent, then one speaks of a simplicial cone.


\begin{proposition}
\label{prop:hull}
	 For $\Ll$ a simplicial cone $\Omega_\mathcal{L}$ is homeomorphic to $\tilde{\Omega}:=[0,1)^d\cup \{*\}$, the one-point compactification of a half-open parallelepiped.
\end{proposition}
\begin{proof}
	W.l.o.g. we may assume that $\Ll=(\RM_{\geq 0})^d$ since any simplicial cone takes that form after a change of coordinates. Choose an increasing homeomorphism $t\colon [0,1) \to  \{-\infty\}\cup  \RM$ and define
	$$\imath\colon  x\in [0,1)^d \;\mapsto\;[t(x_1), \infty) \times ...\times [t(x_d), \infty) \in \mathscr{C}(\RM^d)$$
	and $\imath(*)=\emptyset$ the empty set. The map $\imath\colon \tilde{\Omega}\to \mathscr{C}(\RM^d)$ is clearly injective and the interior $(0,1)^d$ is mapped bijectively to $\Oo_{\RM^d}(\Ll)$.  Moreover, $\imath$ is continuous since its image consists of cartesian products of intervals and it is easy to verify sequential continuity using the well-known convergence criteria for sequences in the Fell topology (e.g. \cite[E.1.2.]{BenedettiPetronio}). From the compactness of $\tilde{\Omega}$, one concludes that $\imath$ is a homeomorphism onto its image. This implies that $\imath$ is a homeomorphism onto $\Omega_\Ll$ because the image of $\imath$ has $\Oo_{\RM^d}(\Ll)$ as a dense subset. 
\end{proof}

One can give $\tilde{\Omega}$ the structure of a CW-complex such that each $n$-face corresponds to an $n$-cell, which also gives $\Omega$ the structure of a CW-complex. By reducing to a subcomplex this also applies more generally to non-pointed cones which are intersections of $m\leq d$ half-spaces:

\begin{proposition}\label{prop: CW}
	Let $\{\lambda_i\}_{i=1}^m\subset \RM^{d}$ be a set of normalized linearly independent vectors and $\mathcal{L}=\bigcap_{i} \mathcal{L}_{\lambda_i}$. For $I\in \Pp(\{1,...,m\})$, the power set, denote $\Ll_I:= \cap_{i\in I} \Ll_{\lambda_i}$ with $\Ll_\emptyset= \RM^d$. Then 
	\begin{equation}\label{eq: hull1}
		\Omega_\mathcal{L}\;=\; \bigsqcup_{I\in \Pp(\{1,...,d\})} \mathscr{O}_{\RM^d}(\mathcal{L}_I)\sqcup\{\emptyset\}.
	\end{equation}
	Moreover, $\Omega_\Ll$ is a CW-complex in such a way that each $\Oo_{\RM^d}(\Ll_I)$ is an $n$-cell with $n=\abs{I}$ via a homeomorphism
	$$x \in \mathrm{span}(\lambda_i)_{i\in I} \;\mapsto\; \Ll_I - x.$$
	There is a linear homeomorphism of $\mathrm{span}(\lambda_i)_{i\in I}$ to $\RM^n$ such that the translation and scaling action on $\mathscr{C}(\RM^d)$ become an affine action on $\RM^n$ as in Definition~\ref{def:flat_cw} with $\Lambda$ a matrix with columns $(\lambda_i)_{i\in I}$.
	
\end{proposition}
\begin{proof}
W.l.o.g. one can assume $m=d$ and reduce to the subcomplex generated by the intersection of the appropriate half-spaces. The cells $\Oo(\Ll_I)$ are precisely the images of parts of the boundary of $[0,1)^d$ under the map $\imath$ from the proof of Proposition~\ref{prop:hull} taking into account the necessary linear transformation.
\end{proof}

\begin{figure}[t!]
	\center
	\includegraphics[width=0.6\textwidth]{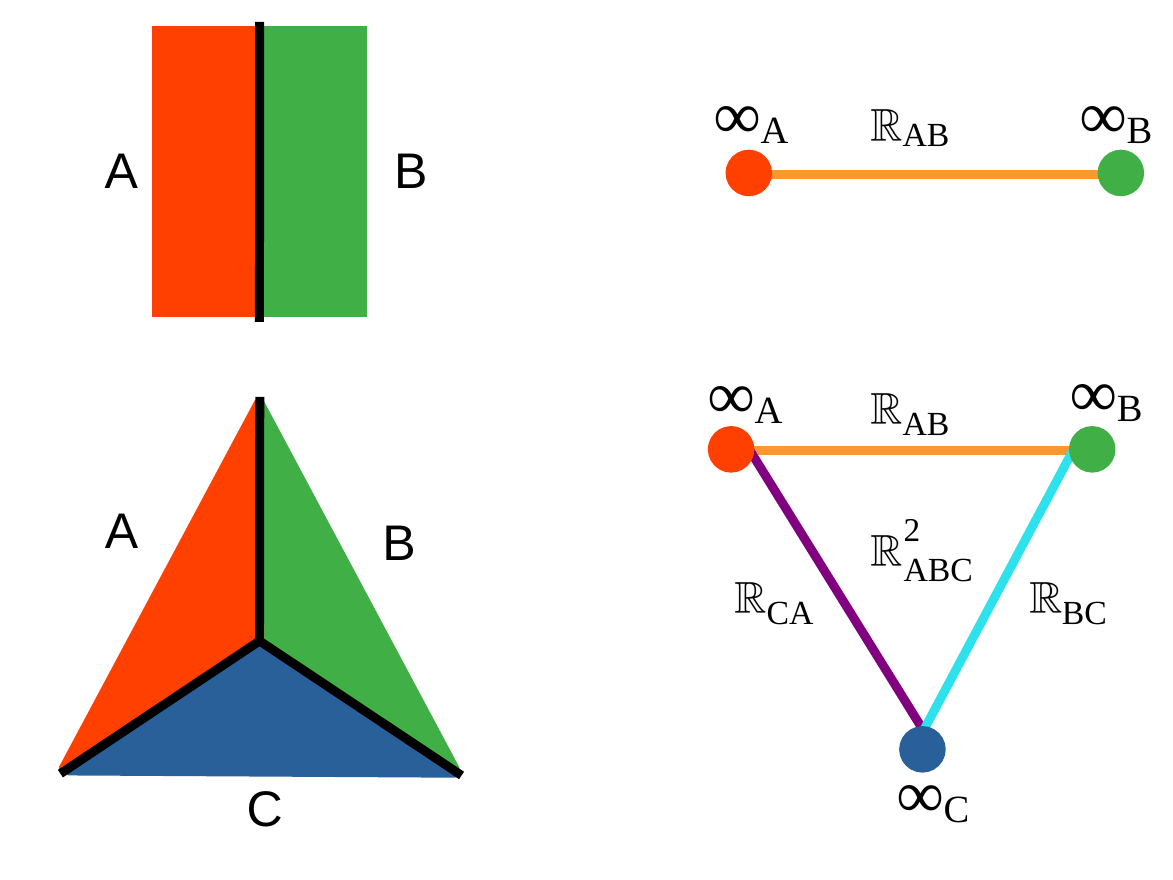}\\
	\caption{\small Examples of real-space geometries and configuration spaces. Top: The configuration space for an interface between two materials A, B consists of a single one-cell $\RM_{AB}$ with two infinite points attached $\Omega=\RM_{AB}\sqcup\{\infty_A\}\sqcup\{\infty_B\}$. Bottom: The configuration space for a Y-type interface between materias A, B and C, has the cell decomposition shown at the bottom right. The interface is invariant only w.r.t. translations orthogonal to the displayed plane, hence corresponding to a codimension $2$ defect which gives rise to a $2$-cell $\RM^2_{ABC}$. The one-cells $\RM_{XY}$ each correspond to the asymptotic interfaces of two materials only and the three points $\infty_A, \infty_B, \infty_C$ of the $0$-cell are the different bulk limits.
	}
	\label{fig:configuration_spaces}
\end{figure}

\begin{remark}
{\rm That $\Omega_\Ll$ can be rather naturally identified with the one-point compactification of $\Ll$ itself is related to the fact that a simplicial cone is self-dual; from computations of Wiener-Hopf compactifications for more general cones \cite{MuhlyRenault,Alldridge,Nica} the structure of the hull as a flat CW-complex is more naturally described in terms of the face lattice of the dual cone.}
\end{remark}

Assuming that the vectors $\lambda_1,...,\lambda_m$ are multiples of lattice vectors then the configuration space is flat and rational.
Let us discuss how the representations $\pi_\omega(h)$ for a self-adjoint Hamiltonian $C(\Omega_\Ll)\rtimes_\alpha \ZM^d$ look in real-space. For $\omega = \Ll-x \in \Oo_{\RM^d}(\Ll)$, which lies in the cell of top degree in the CW-complex, the representation describes two asymptotically translation-invariant materials which are interpolated across an interface layer which looks like the boundary of $\Ll$ translated by a vector $x$. The representations for $\omega$ in the different lower-dimensional cells describe different asymptotic limits, corresponding to cones which are intersections of fewer half-spaces. Note finally that the orbit $\overline{\Oo_{\ZM^d}(\omega)}$ for each $\omega\in \Oo_{\RM^d}(\Ll)$ decomposes precisely as in \eqref{eq:defect_decomp} with the groups $\Gamma_{m,j}$ iterating over the isotropy groups of the cones $\Ll_I$ under the $\ZM^d$-action. 
\begin{remark}
{\rm 
Configuration spaces for interfaces of more than two materials, as in Figure \ref{fig:configuration_spaces} can be obtained from partitions of $\RM^d$ into a finite number of polyhedral regions $\Ll_1,\dots ,\Ll
_k$. We decorate those connected components by assigning labels $\tilde{\Ll}_i:=\{i\}\times \Ll_i$ and consider $\Ll:=\cup_{i=1}^k\tilde{\Ll}_i \in \mathscr{C}(\{1,...,k\}\times \RM^{d})$. There is an obvious action of $\RM^d$ on $ \mathscr{C}(\{1,...,k\}\times \RM^{d})$ and one can again consider the closure in the Fell topology. The resulting configuration space is again a CW-complex in which the $0$-skeleton consists of $k$ points representing the different bulk limits.  \hfill   $\blacktriangleleft$ }
\end{remark}
\section{Connes-Thom isomorphism at work}

The Connes-Thom isomorphism \cite{Connes81} is a natural isomorphism between the K-theory of a crossed product algebra $\Cc\rtimes \RM^n$ with that of $\Cc$. It comes into play due to the following key observation:

\begin{proposition}
\label{prop:affine_crossedproduct}
If $\RM^n$ is equipped with the affine $\RM^d$ action 
$$\RM^d \times \RM^n \ni (x,y) \mapsto x\triangleright y = y + \Lambda^T x = y+ \begin{pmatrix}
    \lambda_1 & \dots & \lambda_n
\end{pmatrix}^T x,$$
with column vectors $\lambda_1,\dots,\lambda_n\in \RM^d$, then one can rewrite
$$C_0(\RM^n) \rtimes_{\alpha} \ZM^d\;\simeq \;C(\TM^d)\rtimes_{\hat{\alpha}} \RM^n,$$
with the $\RM^n$-action $\hat{\alpha}$ induced by the action $\RM^n\times \TM^d(y,k) \mapsto k-\Lambda y$ on $\TM^d = \RM^d/\ZM^d$.
\end{proposition}
\begin{proof}
 The crossed product $C(\TM^d)\rtimes_{\hat{\alpha}} \RM^n$ is the non-unital $C^\ast$-algebra defined by norm-closure of the linear span of operators of the form $\rho(f)\int_{\RM^n} g(x) U_x \difd{x}$ on $L^2(\TM^d)\otimes L^2(\RM^n)$ \cite{Raeburn88} where $f\in C(\TM^d)$, $g\in C_c(\RM^n)$ and the so-called regular covariant representation is 
\begin{align*}
(\rho(f)\psi)(k,y) &\;=\; f(k+\Lambda y) \psi(k,y)\\
(U_x\psi)(k,y) &\;=\; \psi(k,y-x)\qquad \forall \ k\in \TM^d, \ x,y\in \RM^n.
\end{align*}
One can write $\int_{\RM^n} g(x) U_x \difd{x}= \hat{g}(p)$, where  $\hat{g}\in C_0(\RM^n)$ is the Fourier transform of $g$ and $p:=(p_1,...,p_n)$ are the self-adjoint generators of the unitary translation group $(U_x)_{x\in \RM^d}$. Since the unitary operators $u^q= \rho(k\in \TM^d \mapsto e^{\imath k \cdot q})$ with $q\in \ZM^d$,  are dense in $\rho(C(\TM^d))$, one gets that a norm-dense subset of $\rho\rtimes U\big(C(\TM^d)\rtimes_{\hat{\alpha}}  \RM^n\big)$ is given by finite series
$$\sum_{q\in \ZM^d} u^q \hat{g}_q(p), \quad \hat{g}_q\in C_0(\RM^n).$$
Note the commutation relation $u^q U_{x} u^{-q}= e^{\imath x(\Lambda^T q)}U_x$ implies that $u^q$ acts on functions $\hat{g}(p)$ by translation of the argument by $\Lambda^T q$. Due to the universal property of crossed products \cite{Raeburn88}, this commutation relation shows that there is a homomorphism  $C_0(\RM^n)\rtimes_{\hat{\alpha}} \ZM^d \to C(\TM^d)\rtimes_{\hat{\alpha}}  \RM^n$. It clearly has dense range and, since the subalgebra $C_0(\RM^n)$ is embedded isometrically, it is injective (after all, the conditional expectation $C_0(\RM^n)\rtimes_{\alpha} \ZM^d\to C_0(\RM^n)$ is faithful).
\end{proof}
In the following we will write the dual crossed product from the Proposition as $C(\TM^d)\rtimes_{\Lambda}\RM^n$ to specify the underlying affine action. One obtains a continuous field whose fibers can be written as
$$\Aa_t\; :=\; C(\TM^d)\rtimes_{t \Lambda^T}\RM^n :=\; C(\TM^d)\rtimes_{t \Lambda}\RM^n.$$
\begin{proposition}
\label{prop:ktheory_contsection}
For each $t$, there is a Connes-Thom isomorphism
\begin{equation}\label{Eq:CT1}
\partial_{t\Lambda}: K_{i-n}( C(\TM^d)) \to K_i(C(\TM^d)\rtimes_{t \Lambda} \RM^n).
\end{equation}
A continuous section $(\xi_t)_{t\in T}$, $\xi_t\in M_N(C(\TM^d)\rtimes_{t \Lambda}\RM^n)$ of projections/unitaries defines classes $[\xi_t]_i\in K_i(C(\TM^d)\rtimes_{t \Lambda}\RM^n)$ and $(\partial_{t\Lambda})^{-1}([\xi_t]_i)\in K_{i-n}( C(\TM^d))$ does not depend on $t$.
\end{proposition}
\begin{proof}
    The Connes-Thom isomorphism is a natural map $K_i(\Cc)\simeq K_{i+1}(\Cc\rtimes\RM)$  \cite{Connes81} and by composing $n$ Connes-Thom isomorphisms one has $K_i(\Cc)\simeq K_{i+n}(\Cc\rtimes\RM^n)$. The total algebra of the field $\Aa\simeq ( C(T)\otimes C_0(\RM^n))\rtimes_{\alpha^{(T)}}\ZM^d$ can also be rewritten as a crossed product $\Aa\simeq (C(T)\otimes C(\TM^d))\rtimes_{\hat{\alpha}^{(T)}} \RM^n$. There is therefore also a Connes-Thom isomorphism of the entire field
\begin{equation}\label{Eq:CT2}
\partial^{(T)}_{CT}\colon K_{i-n}(C(T)\otimes C(\TM^d)) \to K_i((C(T)\otimes C(\TM^d))\rtimes_{\hat{\alpha}^{(T)}} \RM^n).
\end{equation}
Due to the naturalness property of the Connes-Thom map under equivariant homomorphism \cite[Axiom 2]{Connes81}, restriction to a fiber $t$ supplies a commuting square between \eqref{Eq:CT2} and \eqref{Eq:CT1}, with surjective vertical arrows. One can therefore represent the pre-images $(\partial_{t\Lambda})^{-1}([\xi_t]_i)\in K_{i-n}( C(\TM^d))$ by a norm-continuous family $(\zeta_t)_{t\in T}$ in $M_N(C(\TM^d))$ representing the pre-image $[\zeta]_{i-n}=(\partial_{CT}^{(T)})^{-1}([\xi]_i)$. Hence they all define the same $K$-theory class.
\end{proof}

For the special case of discrete crossed products over $C_0(\RM^n)$, it is also known \cite{ENN93} that the composition $\partial_\Lambda \circ (\partial_{0\cdot\Lambda})^{-1}$ is precisely the natural homomorphism $K_i(\Aa_0)\to K_i(\Aa_t)$ mentioned in Section~\ref{sec-adiabatic}, which in this case is therefore an isomorphism. For an affine configuration space, one will always be able to reduce the calculations of defect Chern numbers from Proposition~\ref{prop:chern_rational} to one on the torus, which paves the way to relate the invariants with Chern numbers of an adiabatic symbol.

\section{Chern numbers of adiabatic quantizations}
\label{sec:chern}

All elements of $K_i(C(\TM^d))$ can be distinguished by numerical pairings with a set of cyclic cocycles, called Chern cocycles. Throughout we identify the torus with $\TM^d=\RM^d /\ZM^d$, therefore there is for any direction $v\in \RM^d$ a $*$-derivation $\nabla_v: C^1(\TM^d)\to C(\TM^d)$ given by $$\nabla_vf\;:=\; \sum_{i=1}^d v_i \nabla_i f.$$
For any tuple $v=(v_1,\dots,v_m)$ of vectors in $\RM^d$ and $(m+1)$-tuple $(f_0,\ldots,f_m)$ of functions $f_j\in C(\TM^d)$ one defines the Chern cocycle
\begin{equation}
\Ch_{v}(f_0,\dots,f_m) 
\;=\; \,
\sum_{\rho \in S_m} (-1)^\rho\, \Tt\big(f_0 \nabla_{v_{\rho(1)}} f_1 \ldots  \nabla_{v_{\rho(m)}} f_m\big)
\;,
\end{equation}
with $\Tt$ the usual trace on $C(\TM^d)$ normalized such that $\Tt(\one)=1$. As a consequence, all non-trivial pairings will exactly have $\ZM$ as their range whenever $v$ is made up of standard basis vectors of $\RM^d$ (see Appendix \ref{app:ktheory}). 
\begin{proposition}
\label{prop:Chern_torus}
Consider all tuples of unit vectors $v_I = (e_{i_1},\dots,e_{i_n})$, $i_1<i_2<\dots<i_m$ with subsets $I=\{i_1,\dots,i_m\}$ of elements in $\{1,\dots,d\}$. We say a class $\xi_{I}\in K_{i}(C(\TM^d))$ is dual to $\Ch_{v_I}$ if 
\begin{equation}
\langle [\xi_{I}]_i, \Ch_{v_J}\rangle \;=\;
	\delta_{I,J}
\end{equation}
for all subsets $I,J\subset \{1,\dots,d\}$  whose length has the same parity as $i$.\\
A dual element $\xi_I$ exists for each subset $I\subset\{1,..,d\}$ and the even respectively odd subsets provide a basis for $K_0(C(\TM^d))\simeq \ZM^{2^{d-1}}$ and $K_0(C(\TM^d))\simeq \ZM^{2^{d-1}}$ respectively. 
\end{proposition}
We recall the proof in Appendix~\ref{app:ktheory}.

Let us now consider the basic building block $C_0(\RM^n)\rtimes_{t \Lambda^T} \ZM^d \simeq C(\TM^d)\rtimes_{t\Lambda}\RM^n$ with affine action as given in Proposition~\ref{prop:affine_crossedproduct}. A basis of the K-groups $K_i(C_0(\RM^n)\rtimes_{t \Lambda^T} \ZM^d)\simeq K_{i+n}(C(\TM^d))$ is given by applying the Connes-Thom isomorphism to the generators $\xi_I$. It is not so easy to write down convenient expressions for those images, however, one can understand them in terms of their Chern numbers via duality Theorems. In the special case where the $\RM^n$-action is trivial there is a rather simple expression: 
\begin{theorem}
\label{th:duality1}
At $t=0$, one has $C(\TM^d)\rtimes_{0\cdot \Lambda} \RM^n\simeq C(\TM^d)\otimes C_0(\RM^n)$. We can define directional derivatives $\nabla_v$ in directions $v\in \RM^{d}\oplus \RM^n$. For a tuple $(v_1,\dots,v_r)$ of such directions define a cyclic $r$-cocyle on  $(C^1(\TM^d)\otimes C_c^1(\RM^n))^{\times (r+1)}$ by
\begin{equation}
\label{eq:suspension_cocycle}
\widehat{\Ch}_{v}(f_0,\dots,f_r) 
\;=\; 
\,
\sum_{\rho \in S_r} (-1)^\rho\, \hat{\Tt}\big(f_0 \nabla_{v_{\rho(1)}} f_1 \ldots  \nabla_{v_{\rho(r)}} f_r\big)
\;,
\end{equation}
with $\hat{\Tt}$ the trace induced by the Haar measure on $\TM^d \times \RM^n$. For directions $v_1,...,v_m$ on the torus and the unit directions $e_1,...,e_n$ of $\RM^n$ this coincides with the so-called cup product
$$\widehat{\Ch}_{v_1, \dots , v_m , e_1 ,\dots , e_n } \;=\; {\Ch}_{v_1, \dots , v_m} \# \Ch_{\RM^n}.$$
With respect to the Connes-Thom isomorphism $\partial_{0\cdot\Lambda}: K_i(C(\TM^d))\to K_{i+n}(C(\TM^d)\otimes C_0(\RM^n))$, one has the duality
$$\langle [\xi]_i, \Ch_{v_1, \dots , v_m}\rangle \;=\;  (-1)^{\frac{1}{2}n(2m+n-1)}\langle \partial_{0\cdot\Lambda}([\xi]_i), {\Ch}_{v_1, \dots , v_m} \# \Ch_{\RM^n}\rangle.$$
\end{theorem}
\begin{proof}
   In the case $n=1$ it is known that 
 $$ \langle [\xi]_i, \Ch_{v_1, \dots , v_m}\rangle \;= \;(-1)^m\langle [\xi]_i, \Ch_{v_1, \dots , v_m}\# \Ch_{\RM}\rangle$$
 with the Connes-Thom isomorphisms for trivial $\RM$-action, see e.g. \cite[Theorem 4.4.6-7]{SSt}. The general version of that duality theorem can be iterated to yield the case $n>1$ with
$$\langle [\xi]_i, \Ch_{v_1, \dots , v_m}\rangle \;=\;  (-1)^{\sum_{k=m}^{n+m-1}k}\langle \partial_{0\cdot\Lambda}([\xi]_i), {\Ch}_{v_1, \dots , v_m} \# \Ch_{\RM^n}\rangle.$$
The sign factor can be simplified as given.  
\end{proof}

For general $\Lambda$, we use a different duality:
\begin{theorem}
\label{th:duality2}
For $v_1,\dots,v_m\in \RM^d$, define a Chern cocycle on $C(\TM^d)\rtimes_{\Lambda} \RM^n$ by
\begin{equation}
\label{eq:lambda_cocycle}
{\Ch}^{\Lambda}_{v}(f_0,\dots,f_m) 
\;=\; 
\,
\sum_{\rho \in S_m} (-1)^\rho\, \hat{\Tt}_\Lambda\big(f_0 \nabla_{v_{\rho(1)}} f_1 \ldots  \nabla_{v_{\rho(m)}} f_m\big)
\;,
\end{equation}
where the dual trace $\hat{\Tt}_\Lambda$ is given by a natural trace on $C_0(\RM^n)\rtimes_{\Lambda^T} \ZM^d \simeq C(\TM^d)\rtimes_{\Lambda} \RM^n$, namely the composition of the conditional expectation $\sum_{q\in \ZM^d} f_q u^q \mapsto f_0 \in C_0(\RM^n)$ with the Lebesgue integral. Then one has a duality of the Chern cocycles in the form
$$\langle \partial_\Lambda [\xi]_i, {\Ch}^\Lambda_{v_1, \dots , v_m}\rangle \;=\; (-1)^{\frac{1}{2}n(2m+n-1)} \langle [\xi]_i, {\Ch}_{v_1 , \dots , v_m , \lambda_1, \dots , \lambda_n}\rangle.$$
\end{theorem} 
\begin{proof}
   This follows from \cite[Theorem 4.5.3]{SSt}, which applies to crossed products by $\RM$ and yields immediately the case $m=1$, which is written as
$$\langle \partial_\Lambda [\xi]_i, {\Ch}^\Lambda_{v_1, \dots , v_m}\rangle = (-1)^m\langle [\xi]_i, {\Ch}_{v_1 , \dots , v_m , \lambda_1}\rangle$$
but can be iterated since the Connes-Thom isomorphism decomposes as an iteration of $n$ one-dimensional Connes-Thom isomorphisms for the $n$ commuting affine translation actions of $\RM$ on $\TM^d$ in the directions $\lambda_1$ to $\lambda_n$. 
\end{proof} 


We can now state our main result on the Chern numbers of adiabatic quantizations:
\begin{theorem}
\label{theorem:adiabatic_chern_numbers}
Let $(\xi_t)_{t\in T}$ be a continuous section of projections/unitaries defining classes in $K_i(C_0(\RM^n)\rtimes_{t\Lambda^T} \ZM^d)$. For $t>0$, $\Lambda=\begin{pmatrix}
\lambda_1 & \dots & \lambda_n
\end{pmatrix}$ and $v_1,\dots,v_m \in \RM^d$ one has 
\begin{equation}
\label{eq:adiabatic_chern}
\langle [\xi_0]_i, \Ch_{v_1,\dots, v_m , \lambda_1 , \dots , \lambda_n}  \# \Ch_{\RM^n}\rangle = t^{-n}\langle [\xi_t]_i, \Ch^{t\Lambda}_{v_1,\dots, v_m}\rangle
\end{equation}
\end{theorem}
\begin{proof}
   By Theorem~\ref{th:duality2} the right hand side of \eqref{eq:adiabatic_chern} is up to the sign factor equal to 
$$
t^{-n}\langle (\partial_{t\Lambda})^{-1})([\xi_t]_i), \Ch_{v_1,\dots, v_m , t\lambda_1 , \dots , t\lambda_n}\rangle\;=\;
\langle (\partial_{t\Lambda})^{-1})([\xi_t]_i), \Ch_{v_1,\dots, v_m , \lambda_1 , \dots , \lambda_n}\rangle$$
and by, Theorem~\ref{th:duality1}, the left hand side of \eqref{eq:adiabatic_chern} is up to the sign  equal to
$$\langle (\partial_{0\cdot\Lambda})^{-1}([\xi_0]_i), \Ch_{v_1,\dots, v_m , \lambda_1 , \dots , \lambda_n}\rangle.$$
For a continuous section, one has by Proposition~\ref{prop:ktheory_contsection}
$$\langle (\partial_{0\cdot\Lambda})^{-1}([\xi_0]_i), \Ch_{v_1,\dots, v_m , \lambda_1 , \dots , \lambda_n}\rangle\;=\;\langle (\partial_{t\cdot\Lambda})^{-1}([\xi_t]_i), \Ch_{v_1,\dots, v_m , \lambda_1 , \dots , \lambda_n}\rangle,$$ 
and the statement follows.\end{proof}

If we want to compute the pairing of one of the Chern cocycles ${\Ch}^\Lambda_{v_1,\dots, v_m}$ over the non-commutative algebra $C_0(\RM^n)\rtimes_{\Lambda^T} \ZM^d$ (the case $t=1$), then we can do so by extending a given projection/unitary to a continuous section of the field, evaluate it at $0$ and then compute an $(m+2n)$-cocycle over the commutative algebra $C_0(\RM^n)\otimes C(\TM^d)$, which is usually much more tractable. 

The keen reader will observe that the cocycles $\Ch^\Lambda_{v_1,\dots, v_m}$ are not enough to tell apart all elements of $K_i(C_0(\RM^n)\rtimes \ZM^d)$; since they are dual to $\Ch_{v_1,\dots, v_m,\lambda_1,...,\lambda_n}$ they can result in a non-trivial pairing only when $v_1,\dots,v_m$ are linearly independent both from each other and $\lambda_1,\dots,\lambda_n$. However, the non-trivial cocycles are precisely enough to label the part of the K-group that survives the localization to a single orbit as in Proposition~\ref{prop:chern_rational}. Theorem~\ref{theorem:adiabatic_chern_numbers} together with Proposition~\ref{prop:quantization} therefore almost finishes the proof of Theorem~\ref{th:main}. It remains to relate the Chern numbers $\Ch^\Lambda$, which are defined over the crossed product $C_0(\RM^n)\rtimes_{\Lambda^T} \ZM^d$, to the cocycles $\Ch^{\omega}$ for a localization to a single orbit $C_0(\Oo(\omega))\rtimes \ZM^d$.

\medskip

We shall use the notation $0_n$ for the zero element of $\RM^n.$ For an affine and rational action in the sense of Definition~\ref{def:flat_cw}, the orbit $\mathcal{O}=\mathscr{O}(0_n)$ of $0_n$ under the $\ZM^d$-action is a lattice in $\RM^n$, i.e. a discrete subgroup with finite co-volume.  As a fundamental domain for the action, one can use any half-open parallelepiped $P\subset \RM^n$ such that $\RM^n = P + \mathcal{O}$ with a unique decomposition for each element.  The orbit $\Oo(\omega)=\omega +\mathcal{O}$ is a translation of $\mathcal{O}$ is the same lattice shifted by $\omega\in \RM^n$. With this information in mind, we can now relate the Chern numbers:
\begin{proposition}
\label{prop:trace_localization}
If the action on $\RM^n$ is affine and rational, then for each $a\in C_0(\RM^n)\rtimes_{\Lambda^T} \ZM^d$ which is trace-class w.r.t. $\hat{\Tt}_\Lambda$ one has
$$\hat{\Tt}_\Lambda(a) \;=\; \int_{\omega \in P} \Tt_\omega(\pi_\omega(a)) \mathrm{d}\omega$$
 with $\Tt_\omega$ the trace described in Proposition \ref{prop:chern_rational}. Furthermore, for each $\omega\in \RM^n$ and $v_1,...,v_m \in \Gamma_\omega$ one has
\begin{equation}
\label{eq:fiberwise_chern}
\langle [\xi]_i, \Ch^\Lambda_{v_1,..., v_m}\rangle \;=\; {\mathrm{Vol}(P)} \,\langle [\pi_\omega(\xi)]_i, \Ch^\omega_{v_1,..., v_m}\rangle
\end{equation}
for any $[\xi]_i \in K_i(C_0(\RM^n)\rtimes_{\Lambda^T} \ZM^d)$. Here, ${\rm Vol}(P)$ is the Lebesgue measure of $P$.
\end{proposition}
\begin{proof}
     The trace on $C_0(\RM^n)\rtimes_{\Lambda^T} \ZM^d$ is given by the conditional expectation $a=\sum_{q\in \ZM^d} a_q u^q \mapsto a_0 \in C_0(\RM^n)$ composed with the Lebesgue integral. On the other hand,
\begin{align*}
\hat{\Tt}_\Lambda(a) &\;=\; \int_{\RM^n}\langle \delta_0, \pi_{\omega}(a)\delta_0\rangle \mathrm{d}\omega\;=\;\int_{P}\sum_{x\in \mathcal{O}}\langle \delta_0, \pi_{x\triangleright \omega}(a)\delta_0\rangle \mathrm{d}\omega\\
&\;=\;\int_{P}\Tt_\omega(\pi_{\omega}(a)) \mathrm{d}\omega
\end{align*}
where recall that $\Tt_\omega$ is the trace on $C_0(\Oo(\omega))\rtimes \ZM^d\simeq C_0(\mathcal{O})\rtimes \ZM^d$ induced by the counting measure on $\mathcal{O}$. Since $\Omega=\RM^n$ is the orbit of $0_n\in \RM^n$ under the $\RM^d$-action there is for any $\omega\in \RM^n$ some preimage $x\in \RM^d$ under $\Lambda^T$ such that $\pi_\omega = \pi_{0_n} \circ \alpha_x$ and $\alpha_x$ is an automorphism. For $d>n$ the preimage is non-unique and the algebras $C_0(\Oo(\omega))\rtimes \ZM^d$ for different $\omega\in \RM^n$ are therefore isomorphic to each other, but not canonically so. Nevertheless, one can choose a continuous function $g: P\to \RM^d$ such that $\pi_\omega=\pi_{0_n}\circ \alpha_{g(\omega)}$ for all $\omega \in P$ in the fundamental domain and in this way identify any family $(\pi_\omega(a))_{\omega\in P}$ with a norm-continuous family of operators in $\pi_{0_n}(C_0(\mathcal{O})\rtimes \ZM^d)$ (since $\alpha$ has norm-continuous orbits). This identification preserves traces and derivatives, therefore also the pairings with the Chern cocycles. In conclusion, the right-hand side of \eqref{eq:fiberwise_chern} does not depend on $\omega\in P$ since all $\pi_\omega(x)$ correspond to the same class in $K_i(C_0(\mathcal{O})\rtimes \ZM^d)$ by homotopy. Since the identification of $P$ with $\RM^n/\mathcal{O}$ was arbitrary, this must be actually true for all $\omega\in \RM^n$.\\
The identity for the trace implies 
$$\langle [x]_i, \Ch^\Lambda_{v_1,..., v_m}\rangle\; =\; \int_P \langle [\pi_\omega(x)]_i, \Ch^\omega_{v_1,..., v_m}\rangle \difd y \;=\; {\mathrm{Vol}(P)}\, \langle [\pi_\omega(x)]_i, \Ch^\omega_{v_1,..., v_m}\rangle$$ for any $\omega \in \RM^n$ since it is the integral over a constant function. 
\end{proof}

The trace $\Tt_\omega$ has the natural normalization such that it corresponds precisely to the trace on $C_0(\Oo(\omega))\rtimes \ZM^d\simeq \KM(\ell^2(\ZM^d/\Gamma_\omega))\rtimes \Gamma_\omega$ induced by composing the conditional expectation $\KM(\ell^2(\ZM^d/\Gamma_\omega))\rtimes \Gamma_\omega\to \KM(\ell^2(\ZM^d/\Gamma_\omega))$ with the usual Hilbert space trace. This choice makes it easy to find Chern cocycles for which the range of the pairing is integer-valued:
\begin{proposition}
In the setting of Proposition~\ref{prop:chern_rational} choose a basis $\ZM^d=\langle b_1,...,b_d\rangle$ such that $\Gamma_\omega = \langle b_1,...,b_{d-n}\rangle$ and denote its dual basis in $\RM^d$ by $\hat{b}_1,...,\hat{b}_d$. Every increasing tuple $I=(i_1,...,i_m)$, $i_k\in \{1,...,d-n\}$ determines a tuple $v_I:=\{\hat{b}_{i_1},...,\hat{b}_{i_m})$.\\
For $d-n>0$ the pairings
$$[\xi]_i \in K_i(C_0(\Oo(\omega))\rtimes \ZM^d)\; \mapsto\; (\langle [\xi]_i, \Ch^\omega_{v_I}\rangle)_{
	\substack{I \subset \{1,...,d-n\} \\ \abs{I}=i\; \mathrm{mod}\,2 }}\; \subset \;\ZM^{2^{d-n-1}}$$
are bijections and provide the isomorphism $K_i(C_0(\Oo(\omega))\rtimes \ZM^d)\simeq K_i(C(\TM^{d-n}))\simeq \ZM^{2^{d-n-1}}$.
For $d-n=0$ the only non-trivial pairing
$$[\xi]_0 \in K_0(C_0(\Oo(\omega))\rtimes \ZM^d)\; \mapsto \;\langle [\xi]_0, \Ch_\emptyset^\omega\rangle\; =\; \Tt_\omega(\xi) \in \ZM$$
provides the isomorphism $K_0(C_0(\Oo(\omega))\rtimes \ZM^d)\simeq \ZM$.
\end{proposition}
\begin{proof}
    
Recall from Proposition~\ref{prop:chern_rational} that $C_0(\Oo(\omega))\rtimes \ZM^d \simeq \KM(\ell^2(\ZM^d/\Gamma_\omega))\rtimes \Gamma_\omega$. The isomorphism is simply that for $x= x_1 b_1 + ... + x_{d}b_{d}$ any Fourier series with coefficients $f_x\in C_0(\Oo(\omega))$ decomposes as
$$\sum_{x\in \ZM^d} f_x u^{x_1 b_1}\dots u^{x_d b_d} \;=\; \sum_{x\in \ZM^d} \left(f_x u^{x_{d-n+1}b_{d-n+1}}\dots u^{x_{d}b_{d}}\right) (u^{x_1b_1}\dots u^{x_{d-n} b_{d-n}})$$
where operators of the form $f_x u^{x_{d-n+1}b_{d-n+1}}...u^{x_{d}b_{d}}$ generate a subalgebra which commutes with all $u^{x_1b_1}... u^{x_{d-n} b_{d-n}}$ and is isomorphic to $\KM(\ell^2(\ZM^d/\Gamma_\omega))$.\\
For a dual basis vector $v=\hat{b}_i$ in \eqref{eq:derivation} the action of $\nabla_v$ under the isomorphism can be read off from $$\nabla_v (u^\gamma) \;=\; -2\pi \imath x_i u^\gamma, \qquad \forall\, \gamma=x_1 b_1 + ... + x_{d-n}b_{n-d}\in \Gamma_\omega$$ on the unitary generators $u^\gamma$ of the crossed product. Furthermore $\nabla_v$ acts trivially on $\KM(\ell^2(\ZM^d))$. The relation implies that there is an isomorphism $\KM(\ell^2(\ZM^d/\Gamma_\omega))\rtimes \Gamma_\omega\simeq \KM(\ell^2(\ZM^d/\Gamma_\omega))\otimes C(\TM^{d-n})$ with the flat cubic torus $\TM^{d-n}=\RM^{d-n}/\ZM^{d-n}$ such that each $\nabla_{\hat{b}_i}$ becomes precisely one of the coordinate derivatives of $\RM^{d-n}$. Since the cocycles $\Ch_{v_I}$ in Proposition~\ref{prop:Chern_torus} are correctly normalized to give integers and label the K-groups, we conclude the same for $\Ch^\omega_{v_I}$. 
\end{proof}

Having finished the proof of Theorem~\ref{th:main} we now have all the necessary tools to compute the Chern numbers for lattice realizations of adiabatic symbols. For the remainder of this section, we will comment on possible generalizations. The key part of the argument that allows us to describe the K-theoretic quantization homomorphism $K_i(\Aa_0)\to K_i(\Aa_1)$ explicitly is that it is the composition of a Connes-Thom isomorphism and its inverse. This is crucial because there are well-known duality results for the Chern cocycles involved. It appears feasible to generalize this approach to other instances of K-theory where an analogue of the Connes-Thom isomorphism is available, for instance also for real K-theory, which includes various forms such as KO-, KR-, or van Daele K-theory. All of these frameworks can be employed to describe topological insulators with internal symmetries of the Altland-Zirnbauer type \cite{BKR, Kellendonk17, AMZ}. One can then, however, no longer distinguish all generators using the pairing with cyclic cohomology, especially if torsion components are involved. Instead one must use the more general duality pairing with the respective K-homology given by the Kasparov product, which has similar properties. In particular, the homology classes can be built from exterior products of one-dimensional Dirac operators in a similar way that our Chern cocycles are cup products of $1$-cocycles. The boundary map in terms of the Kasparov product leads to similar dualities as we used for the cocycles (see e.g. \cite{BKR, BM19}). This approach could have been used for the complex K-groups that we study in this paper, but it is technically much more complicated than the approach via cyclic cohomology and would involve additional steps to turn the pairings with K-homology into concrete, computable expressions. To some extent, there are also Connes-Thom isomorphisms for equivariant K-theory (e.g.\cite{Kasparov95}), but they take a more complicated form (unless the additional group action commutes with the $\RM^n$-action, but for crystalline symmetries there are many interesting cases where this is not true). How useful those are to explore K-theoretic invariants for equivariant adiabatic quantizations remains to be seen.

\section{Boundary maps in the adiabatic picture}
\label{sec:boundary_maps}

In this section, we assume that $\Omega$ is a CW-complex with skeleton filtration $\Omega_0 \subset \dots \subset \Omega_d=\Omega$ given by closed invariant subsets. For any sort of bulk-boundary or bulk-defect correspondence, one will consider Hamiltonians that are gapped at some asymptotic limit $\Omega_\infty \subset \Omega$. Generally, this $\Omega_\infty$ will take the form $\Omega_\infty=\Omega_{n-1}$ for some $n$, which means that we are looking at Hamiltonians which are gapped in the bulk and at all defects of codimension less or equal to $n-1$, but possibly bind topologically protected states at defects of higher codimension. We then have for each $n \leq d$ a boundary map $$\partial_n\colon K_i(C(\Omega_{n-1})\rtimes_\alpha \ZM^d)\to K_{1-i}(C_0(\Omega_{n}\setminus\Omega_{n-1})\rtimes_\alpha \ZM^d).$$ 
This map can, for general $\Omega_n$ and $\Omega_{n-1}$, often be difficult to describe explicitly due to a lack of explicitly known numerical invariants like the Chern cocycles that one can use to distinguish different K-theory classes. 

We will use the adiabatic picture to obtain at least some and sometimes complete information about this map. Let us therefore now consider the case of a trivial $\RM^d$-action, which is relevant for the endpoint $t=0$ of our continuous fields. The relevant boundary maps now take the form
$$\partial_n\colon K_i(C(\Omega_{n-1})\otimes C(\TM^d))\to K_i(C_0(\Omega_{n}\setminus \Omega_{n-1})\otimes C(\TM^d)).$$
By construction of a CW-complex one can easily compute those once one understands the boundary maps for a single sphere, i.e. those of the exact sequences
$$0 \to C_0(\DM^{n}\setminus \SM^{n-1})\otimes C(\TM^d) \to C(\DM^{n})\otimes C(\TM^d)\to C(\SM^{n-1})\otimes C(\TM^d)\to 0.$$
After all, we can separately consider the disjoint copies of $\RM_j^{n}$ in $\Omega_{n}\setminus \Omega_{n-1}$. From the structure of the CW-complex there is an attaching map $\SM^{n-1}\to \partial\RM_j^{n}\subset \Omega_{n}$, which one can pull back to obtain a  commutative diagram 
\begin{equation}
\begin{tikzcd}
	K_i(C(\Omega_{n-1})\otimes C(\TM^d)) \arrow[d] \arrow[r,"{\partial_n}"] & \arrow[d]  K_{i-1}(C_0(\Omega_{n}\setminus\Omega_{n-1})\otimes C(\TM^d))\\ 
	K_i(C(\partial\RM_j^{n})\otimes C(\TM^d)) \arrow[d] \arrow[r,"{\partial_n}"] & \arrow[d]  K_{i-1}(C_0(\RM_j^{n})\otimes C(\TM^d))\\ 
		K_i(C(\SM^{n-1})\otimes C(\TM^d)) \arrow[r,"{\partial_{\SM^{n-1}}}"] & K_{i-1}(C_0(\DM^{n}\setminus \SM^{n-1})\otimes C(\TM^d))
	\end{tikzcd}
\end{equation}
where the lower right vertical map is an isomorphism. 
\begin{proposition}
\label{prop:boundary_sphere}
Identify $\DM^{n}\setminus \SM^{n-1}$ with $\RM^{n}$ by some homeomorphism and then orient $\SM^{n-1}$ by the outward pointing normal. The boundary map $$\partial_{\SM^{n-1}}\colon K_i(C(\SM^{n-1})\otimes C(\TM^d)) \to K_{1-i}(C_0(\RM^{n})\otimes C(\TM^d))$$
is then uniquely determined by the equality
\begin{equation}
\label{eq:bbc_chern_adiabatic}
\langle \partial_{\SM^{n-1}}[\xi]_i, [\Ch_{\RM^{n}}\# \Ch_{v_1, \dots, v_m}]\rangle\; =\; \langle [\xi]_i, [\Ch_{\SM^{n-1}}\# \Ch_{v_1, \dots, v_m}]\rangle.
\end{equation}
\end{proposition}
\begin{proof}
    This is a simple consequence of the Propositions~\ref{prop:cup_product} and \ref{prop:sphere_boundary_map}.
\end{proof}


Let us finally discuss more generally the range of the adiabatic quantization map. A physical Hamiltonian will live in one of the representations $\pi_\omega$ on $\ell^2(\ZM^d)$ whose image is isomorphic to $C_0(\overline{\Oo(\omega)})\rtimes_\alpha \ZM^d$, with $\overline{\Oo(\omega)}$ the closure of the orbit $\ZM^d\triangleright \omega$ in $\Omega$. We have a commuting diagram of the form
\begin{equation}
\label{eq:comm_diagram}
\begin{tikzcd}
	K_i(C(\Omega_\infty)\otimes C(\TM^d)) \arrow[d, "Q_i"] \arrow[r,"{\partial_{t=0}}"] & \arrow[d,"Q_{i-1}"]  K_{i-1}(C_0(\Omega\setminus\Omega_{\infty})\otimes C(\TM^d))\\ 
	K_i(C(\Omega_\infty)\rtimes_\alpha \ZM^d) \arrow[d, "\pi_\omega"] \arrow[r,"{\partial_{t=1}}"] & \arrow[d, "\tilde{\pi}_\omega"]  K_{i-1}(C_0(\Omega\setminus\Omega_{\infty})\rtimes_\alpha \ZM^d)\\ 
		K_i(C(\overline{\Oo(\omega)}\cap \Omega_\infty)\rtimes_\alpha \ZM^d) \arrow[r,"{\partial_\omega}"] & K_{i-1}(C_0(\overline{\Oo(\omega)}\setminus\Omega_{\infty})\rtimes_\alpha \ZM^d)
	\end{tikzcd}
\end{equation}
where $Q_i$ is the homomorphism defined in \eqref{eq: Q}. The localized boundary map $\partial_\omega$ is important because it only incorporates information about a single lattice Hamiltonian and its asymptotic limits, as opposed to $\partial_{t=1}$ which has as its input a continuously parametrized family of Hamiltonians. Generally, we can obtain complete information about $\partial_\omega$ through the adiabatic approach only if $\pi_\omega \circ Q_i$ is a surjection up to the quotient by the kernel of $\partial_\omega$. An obvious case where this happens is when the $\ZM^d$-action on $\Omega_\infty$ is trivial and $\Omega_\infty\subset \overline{\Oo(\omega)}$, since then $\pi_\omega \circ Q_i$ is the identity map. That special case is relevant, for example, for the configuration space of an elementary codimension $n$ defect mentioned in the introduction and discussed again below, where the action on $\Omega_\infty=\SM^{n-1}$ is trivial.

In general, $\pi_\omega \circ Q_{i-1}$ will not be an injection and it will therefore happen that the quantization of some adiabatic symbols have protected boundary states as an element of $C(\Omega)\otimes C(\TM^d)$ or even $C(\Omega)\rtimes_\alpha \ZM^d$, but which are not topologically protected anymore when localized to a single realization. Those can still be interesting as approximate boundary modes which become more and more robust the closer the system is to the adiabatic limit.

\section{Applications}\label{Sec:App}
This section is devoted to the construction of the explicit examples of Theorem~\ref{th:main}  to demonstrate that by correctly choosing configuration spaces one ends up with lattice Hamiltonians which model non-trivial geometries and have interesting boundary or defect modes.

\subsection{Bulk-interface correspondence}
As our first application, we consider bulk-interface correspondence for two bulk materials with a common spectral gap separated by a dividing hyper-surface. If the two materials have different Chern numbers then there may be in-gap interface states localized to the hyper-surface which carry Chern numbers depending on the differences of the bulk Chern numbers and the angles of the surface. This setting has been studied before both with and without K-theory \cite{Deni1, Kot, Dani} and experts will know that the main result of this section can be pieced together easily from well-known results about bulk-edge correspondence in the literature (esp. \cite{PSbook,SSt}). Nevertheless, we revisit this simple case to demonstrate how the adiabatic approach plays out.

The configuration space for an interface is in the simplest case given by $\Omega=\Omega_1=\RM \sqcup \{-\infty\} \sqcup \{+\infty\}$ with the topology equal to the two-point compactification $\Omega\simeq [-\infty,\infty]$. This is the smallest configuration space which contains Example~\ref{Ex:Interface1} and has an $\RM^d$-action. The asymptotic part as far as the boundary maps are concerned will be the infinite points $\Omega_\infty = \Omega_0 = \{+\infty\}\sqcup \{-\infty\}$ representing the two bulk materials. The action on the one-cell $\Omega_1\setminus \Omega_0$ is affine and given by $$(y,x) \in \RM \times \RM^d \;\mapsto\; y - \Lambda x\;\equiv\;y - \lambda \cdot x$$
for some vector $\lambda\in \SM^{d-1}$ which carries the interpretation of the normal vector of the dividing hypersurface between the materials. To satisfy the rationality condition all components of the vector $\lambda$ must be rational multiples of the same real number. As equivalent conditions this is the case if and only if $\lambda\cdot \ZM^d=c_\lambda \ZM$ for some $c_\lambda>0$, or more geometrically if the lattice truncated to the half-space $\{x\in \ZM^d: \; \lambda\cdot x \geq 0\}$ is invariant under translations by $d-1$ linearly independent lattice vectors.

Since $\Omega_\infty$ consists of two separated invariant points we have $$K_i(C(\Omega_\infty)\rtimes \ZM^d)\simeq K_i(C(\TM^d))\oplus K_i(C(\TM^d))$$ With the evaluation homomorphisms $\ev_\pm: C(\Omega_\infty)\rtimes \ZM^d \to C(\TM^d)$ at $\pm \infty$ those K-theory classes can be labeled by the two sets of Chern cocycles $\Ch\circ \ev_\pm$.

It is interesting to consider here the localization map. Due to the rationality condition, the orbit of any $\omega\in \RM$ is equal to $\omega + c_\lambda \ZM$. The closure of the orbit in $\Omega$ can be identified with a two-point compactification of $\ZM$ adding two points $\{\pm\infty\}$. Therefore, \eqref{eq:comm_diagram} takes the form
\begin{equation}
\begin{tikzcd}
	K_i(C(\Omega_\infty)\otimes C(\TM^d)) \arrow[d, "Q_i"] \arrow[r,"{\partial_{t=0}}"] & \arrow[d, "Q_{1-i}"]  K_{1-i}(C_0(\RM)\otimes C(\TM^d))\\ 
	K_i((C(\Omega_\infty)\rtimes \ZM^d) \arrow[d, "\pi_\omega"] \arrow[r,"{\partial_{t=1}}"] & \arrow[d, "\tilde{\pi}_\omega"]  K_{1-i}(C_0(\RM)\rtimes_\alpha \ZM^d)\\ 
		K_i(C(\Omega_\infty)\rtimes \ZM^d) \arrow[r,"{\partial_\omega}"] & K_{1-i}(C_0(\omega + c_\lambda \ZM)\rtimes_\alpha \ZM^d) \arrow[r,equal] &  K_{i-1}(C(\TM^{d-1})).
	\end{tikzcd}
\end{equation}
In this case $Q_i$ and $\pi_\omega$ are actually the identity map and therefore isomorphisms. The boundary map $\partial_{t=0}$ is a surjection and $Q_{1-i}$ an isomorphism, hence $\partial_{t=1}$ is also a surjection and has the same kernel as $\partial_{t=0}$. The localization map $\tilde{\pi}_\omega$ is a surjection and its kernel is made up of precisely those elements for which the pairings with all of the Chern numbers $\Ch^\omega_{v_1,...,v_m} \circ \tilde{\pi}_\omega$ from Proposition~\ref{prop:chern_rational} vanish. All of this follows from consideration of Chern numbers, since Theorem~\ref{theorem:adiabatic_chern_numbers} and Proposition~\ref{prop:trace_localization} imply:

\begin{proposition}
For any $v_1,...,v_m$ the boundary map ${\partial_{t=0}}$ is determined by
$$
\begin{aligned}
& (-1)^m\langle \partial_{t=0}([\xi]_i), \Ch_{v_1,...,v_m, \lambda}\#\Ch_{\RM}\rangle \\ 
& \qquad \qquad \qquad =\;  \langle [\xi]_i, \Ch_{v_1,...,v_m,\lambda}\circ \ev_+\rangle - \langle [\xi]_i, \Ch_{v_1,...,v_m,\lambda}\circ \ev_-\rangle
\end{aligned}
$$
and therefore ${\partial_{t=1}}$ satisfies $$(-1)^m\langle \partial_{t=1}([\xi]_i), \Ch^\Lambda_{v_1,...,v_m}\rangle\; =\; \langle [\xi]_i, \Ch_{v_1,...,v_m,\lambda}\circ \ev_+\rangle - \langle [\xi]_i, \Ch_{v_1,...,v_m,\lambda}\circ \ev_-\rangle.$$
The localized boundary map $\partial_\omega$ for $\omega\in \Omega\setminus \Omega_\infty \simeq \RM$ is then completely determined by the equality
$$ c_\lambda  (-1)^m\langle \partial_{\omega}([\xi]_i), \Ch^\omega_{v_1,...,v_m}\rangle \;=\; \langle [\xi]_i, \Ch_{v_1,...,v_m,\lambda}\circ \ev_+\rangle - \langle [\xi]_i, \Ch_{v_1,...,v_m,\lambda}\circ \ev_-\rangle.$$
\end{proposition}
\begin{proof}
    The first line follows from Proposition~\ref{prop:boundary_sphere} if we orient $\RM$ positively and use that in the case here $\Ch_{\SM^0}= \ev_+ - \ev_-$. The second line then follows from Theorem~\ref{theorem:adiabatic_chern_numbers}. The third line is Proposition~\ref{prop:trace_localization} and this completely determines the boundary map following from the completeness of those Chern cocycles according to Proposition~\ref{prop:chern_rational}.
\end{proof}

We can now discuss the consequences of those results for interface Hamiltonians that are gapped at the two infinite points. For simplicity, we do this in two dimensions. Assume that $H=H^*\in M_N(C(\Omega)\rtimes \ZM^2)$ has a bulk spectral gap at $0$, \textit{i.e.} $0$ lies in a compact interval $\Delta$ contained in a spectral gap of the bulk Hamiltonian $\ev(H):=(H_+,H_-)$. One associates to $H$ the bulk invariant
$$\Psi_0(H)\;=\;[P_{\leq 0}(H_+)]_0 \oplus [P_{\leq 0}(H_-)]_0\in K_0\big(C(\TM^2)\oplus C(\TM^2)\big)$$
defined by the Fermi projections and the interface invariant
$$\Theta_1(H)\;=\;[e^{\imath \pi f(H)}]_1 \in K_1\big(C_0(\RM)\rtimes \ZM^2\big)$$
where $f\in C(\RM)$ is any function that takes the constant values $-1$ below $\Delta$ and $1$ above $\Delta$. Note that one can write $e^{\imath \pi f(H)}=1+ g(H)$ with a function $g\in C_c(\Delta)$, hence if that unitary is non-trivial then $H$ must have spectrum within $\Delta$. Similarly, the boundary-localized invariant $$\tilde{\pi}_\omega(\Theta_1(H))\;=\;[e^{\imath\pi f(\pi_\omega(H))}]_1\in K_1(C_0(\Oo(\omega))\rtimes \ZM^2)\;\simeq\; K_1(C(\TM))$$ provides an obstruction to having a spectral gap within $\Delta$ for the self-adjoint lattice Hamiltonian $\pi_\omega(H)\in \Bb(\ell^2(\ZM^2)\otimes \CM^N)$.

The boundary invariants are related to the bulk invariants by the maps $\partial_{t=1}$ respectively $\partial_\omega$. There is only one Chern number that can lead to a non-trivial result under $\partial_\omega$, hence it is completely determined by the equality
\begin{equation}
\label{eq:chern_vs_current}
 c_\lambda \langle \partial_\omega(\Psi_0(H)), \Ch^\omega_{v}\rangle\; =\;\langle \Psi_0(H), \Ch_{v, \lambda} \circ \ev_{+}\rangle - \langle \Psi_0(H), \Ch_{v, \lambda}\circ \ev_{-}\rangle
\end{equation}
for any $v\in \RM^2$. If $v$ is a unit vector pointing into a direction orthogonal to $\lambda$ then the r.h.s. of \eqref{eq:chern_vs_current} reduces to the difference of Chern numbers $\Ch_{e_1,e_2}\circ \ev_{\pm}$, up to a sign that depends on the orientation, hence it is an integer. The factor $ c_\lambda $ has a physical interpretation: The trace $\Tt_\omega$ comes from the counting measure on $\lambda\cdot \ZM^2\subset \RM$ with no regard to the density of those points in real-space. Noting that $\lambda\cdot \ZM^2$ is the projection of the lattice onto a line in real-space it becomes apparent that one can define a trace per unit surface area if one renormalizes the measure on $\lambda\cdot \ZM^2$ such that the volume of any large interval $[a,b]\cap (\lambda\cdot \ZM^2)$ is approximately $b-a$. Since that interval has approximately $(a-b) c_\lambda ^{-1}$ integer points the normalized trace per surface area is $\hat{\Tt}_\omega =  c_\lambda \Tt_\omega$. The left-hand side of \eqref{eq:chern_vs_current} equals an integrally quantized conductivity of the interface modes in the direction orthogonal to $\lambda$ and for this physical observable one uses the trace per surface area (see for instance \cite[Remark 4.3]{Dani2}). Hence one incurs the same factor  $ c_\lambda $ when expressing it in terms of $ \Ch^\omega_{v}$ which uses $\Tt_\omega$ instead.

\medskip

Let us now consider a self-adjoint adiabatic symbol $H^{(0)} \in C(\Omega\times \TM^2,M_N(\CM))$ which is invertible at $\Omega_\infty$. A simple example is
$$H^{(0)}(x,k) \;=\;H_{-}(k) f(x) + H_{+}(k) (1-f(x))$$
for a function $f\in C(\Omega)\simeq C([-\infty,\infty])$ which converges to $1$ at $-\infty$ and $0$ at $+\infty$, and two gapped Hamiltonians $H_{\pm}$. Quantizing the symbol turns it into a self-adjoint lattice Hamiltonian which slowly interpolates between $H_\pm$ in real-space. 

If we impose a chiral symmetry $JH_\pm J = -H_\pm$ then the bulk Hamiltonians define classes $\Psi_1(H_\pm)\in K_1(C(\TM^2))\simeq \ZM^2$. Taking an adiabatic quantization $(H^{(t)})_{t\in [0,1]}$ which preserves the chiral symmetry as in Proposition~\ref{prop:quantization}, the bulk-interface correspondence takes the form
\begin{equation*}
\begin{split}
 c_\lambda \langle \partial_\omega(\Psi_1(H^{(1)})), \Ch^\omega_{\emptyset}\rangle\; &=\; \langle \partial_\omega(\Psi_1(H^{(0)})), \Ch_{\lambda}\#\Ch_{\RM}\rangle\\
&=\;\langle \Psi_1(H_+), \Ch_{\lambda}\rangle - \langle \Psi_1(H_-), \Ch_{\lambda}\rangle  
\end{split}
\end{equation*}
The pairing with $\Ch^\omega_\emptyset$ gives the signed number density of bound states to the interface region graded by $J$. Importantly, only the bulk Chern numbers $\Ch_\lambda$ lead to protected bound states in the lattice model, whereas classes which pair non-trivially only with $\Ch_v$ for $v\perp \lambda$ get mapped to zero under $\partial_\omega$. On the other hand,
$$\langle \partial_{t=0}(\Psi_1(H^{(0)})), \Ch_{v}\#\Ch_{\RM}\rangle \;=\;\langle \Psi_1(H_+), \Ch_v\rangle - \langle \Psi_1(H_-), \Ch_{v}\rangle$$
any difference of Chern numbers forces a topologically protected gap-closing of the adiabatic symbol. Because the quantization map $K_1(C(\Omega)\otimes C(\TM^2))\to K_1(C(\Omega)\rtimes \ZM^2$ is an isomorphism, the boundary invariant $\Theta_0(H^{(1)})\in K_1(C(\Omega)\rtimes \ZM^2)$ is also non-trivial, implying that $0$ is in the spectrum of $H^{(1)}$. However, this obstruction only applies to the family $(\pi_\omega(H^{(1)}))_{\omega\in \RM}$ as a whole and does not affect any single lattice realization $\pi_\omega(H^{(1)})$ which one can gap by a small perturbation unless $\pi_\omega(\Theta_0(H^{(1)}))\neq 0$. This sort of topological gap-closing can therefore be an artifact of the adiabatic quantization, rather than a genuine topological effect. Only those gap-closings forced by a non-trivial Chern number $\Ch^\omega_v$ are stable when localized to a single orbit. 


\subsection{Elementary codimension {\it n} defect}
\label{sec:ex_defect}
Consider a $d$-dimensional material with a defect that is invariant under translations in $d-n$ independent directions $b_1$, $...$,$b_{d-n}\in \ZM^d$, but not in the remaining ones. We want to study an elementary version of such a defect, which does not occur in conjunction with any other boundaries, and see when such a defect can bind topologically protected states. Similar to Example~\ref{Ex:PointDefect}, the configuration space in this case shall be homeomorphic to the closed $n$-dimensional disk $\Omega_n = \DM^n \simeq \RM^n \sqcup \SM^{n-1}$ with $(n-1)$-skeleton $\Omega_{n-1}=\SM^{n-1}$ and the action of $\RM^n$ on the $n$-cell $\Omega_n\setminus \Omega_{n-1}\simeq \RM^{n}$ derived from the affine action of $\RM^d$ on the orthogonal complement of $\langle b_1,...,b_{d-n}\rangle$. The boundary sphere $\SM^{n-1}$ represents a \textit{top view} of an asymptotic cylinder $\RM^{d-n}\times \SM^{n-1}$ far away from and surrounding the defect, which one can imagine as being localized at or around a set of the form $\RM^{d-n}\times \{0_n\}\subset \RM^d$ in real-space. Thus $\Omega_{n-1}$ distinguishes different radial limits far away from the defect, while the single $n$-cell labels the relative position of the defect w.r.t. some fiducial point of the lattice. 

On $\Omega_{n-1}$, translations and scaling are represented by the trivial action, whereas on $\Omega_n\setminus \Omega_{n-1}$ they act by affine transformations with directions $\Lambda= (\lambda_1 ... \lambda_n)$, depending on how the defect is embedded into real-space. The CW-complex is therefore flat at level $n$, which is enough for our purposes. It is also rational, since $b_1,...,b_{d-n}$ are lattice vectors. The topology of $\Omega_n$ results from the radial compactification of $\RM^n$ and it is not difficult to see that, for any $\omega\in \Omega_n\setminus \Omega_{n-1}$, the orbit $\ZM^d\triangleright \omega$ is dense in $\Omega_n$. Therefore, we are in the setting mentioned at the end of Section~\ref{sec:boundary_maps} where we can completely determine the boundary map localized to a single orbit:
\begin{proposition}
In the above setting, one has $C(\Omega_{n-1})\rtimes \ZM^d\simeq C(\TM^d)\otimes C(\SM^{n-1})$ and the boundary map
$$\partial_\omega: K_i(C(\SM^{n-1})\rtimes \ZM^d)\mapsto K_{1-i}(C_0(\Oo(\omega))\rtimes \ZM^d)$$
for any $\omega\in \Omega_n\setminus \Omega_{n-1}$ is completely determined by the duality
$$\langle \partial_\omega [\xi]_i, \Ch^\omega_{v_1,...,v_m}\rangle =\frac{(-1)^{\frac{1}{2}n(2m+n-1)}}{\mathrm{Vol}(\RM^n/\ZM^d)}\langle [\xi]_i, \Ch_{v_1,...,v_m, \lambda_1,...,\lambda_n}\# \Ch_{\SM^{n-1}}\rangle.$$
\end{proposition}

For the top class where $m=d-n$ the bulk cocycle becomes $\Ch_{\TM^d}\# \Ch_{\SM^{n-1}}$ and a non-trivial example can always be constructed as a Dirac vector field: choose an irreducible representation $\gamma_1,...,\gamma_{d+n}$ of the Clifford algebra $\CM_{d+n}$ and set
$$(k,x)\in \TM^d \times \SM^{n-1} \mapsto H(k,x) = \sum_{i=1}^{d+n} \gamma_i f_i(k,x)$$
for a function $f:\TM^d \times \SM^{n-1}\to \SM^{d+n-1}\subset \RM^{d+n}$. Such $H$ is gapped and has a chiral symmetry if $d+n$ is even and the Chern number of the Fermi projection respectively unitary is up to a sign equal to the mapping degree of $f$. 

Let us note that in \cite{ProdanJPA2021} one can find a geometric realization of a codimension 2 defect in a two-dimensional material obtained by cutting out a cone from the lattice and gluing the edges back together. It was found numerically that an asymptotic three-dimensional Chern number leads to a bound state at the center of the defect. This is consistent with the results of our formalism; while our defects here do not actually modify the lattice it is likely that the two situations are equivalent on the level of K-theory in the sense that both exact sequences define the same extension class in KK-theory. It would be interesting to make this connection to lattice-modifying geometric defects more explicit, but we must leave that to future work.

\subsection{Corner states}\label{sec:quarter}
We consider here a quarter-space geometry defined by the  cone $ \bigcap_{i=1}^2 \{x\in \RM^d: \lambda_i\cdot x \geq 0\}$, given by the non-empty intersection of two half-spaces with normal vectors $\lambda_1$, $\lambda_2$. As seen in section~\ref{sec:flat}, its configuration space is
$$\Omega\;=\; \Omega_2 \;=\; \RM^2 \sqcup \RM_1 \sqcup \RM_2 \sqcup \{\pm \infty\},$$
and it contains a defect of codimension $2$ localized to the tip of the cone, two defects of codimension $1$ corresponding to the two $(d-1)$-dimensional faces, and two bulk regions at $\pm \infty$. The $\ZM^d$-action on $\RM^2$ is given by affine translation as in Definition~\ref{def:flat_cw} where $\Lambda$ is a matrix with rows $\lambda_1$, $\lambda_2$. For $\RM_i$ the $\ZM^d$-action is implemented by the translation with the scalar product with $\lambda_i$, respectively.

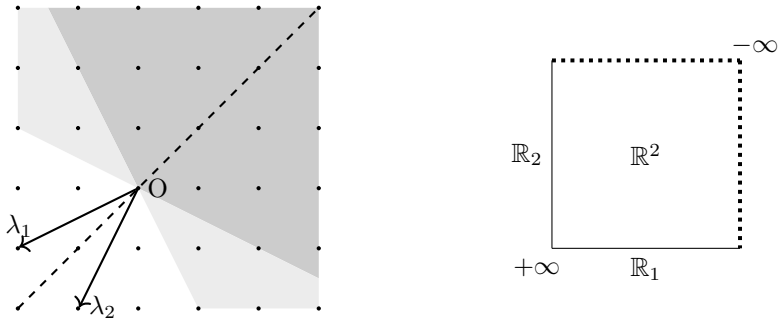
\begin{figure}
	\centering

\begin{tikzpicture}[arrowmid/.style={postaction={decorate},
		decoration={markings,
			mark=at position .6 with {\arrow{latex}}}},]	
 
		\begin{scope}[shift={(1,2.5)},scale=1]
	\coordinate (A) at (0,0);
	\coordinate (B) at (5,0);
	\coordinate (C) at (5,5);
	\coordinate (D) at (0,5);
	
	
	\coordinate (MAB) at ($(A)!0.5!(B)$);
	\coordinate (MBC) at ($(B)!0.5!(C)$);
	\coordinate (MCD) at ($(C)!0.5!(D)$);
	\coordinate (MDA) at ($(D)!0.5!(A)$);
	
	\coordinate (Center) at ($(A)!0.5!(C)$);

    \tikzset{midarrow/.style={postaction={decorate, 
    decoration={markings, mark=at position 0.5 with {\arrow{<}}}}}}
	\draw (Center)  -- (MBC) node[midway, below ] {$\RM_1$};
	\draw (Center) -- (MCD) node[midway,  left] {$\RM_2$};
	
	\node at (2.3,2.25) {$+\infty$};
	
	\node at (3.75,3.75) {$\RM^2$};
	\node at (5.2,5.2) {$-\infty$};
	
	\draw[dotted, line width=1.5pt] (MCD) -- (C);
	\draw[dotted, line width=1.5pt] (C) -- (MBC);
	\end{scope}

\begin{scope}[shift={(-2,5.8)},scale=0.8]
  \fill[gray!40] (0,0) -- (-1.5, 3) -- (3,3) -- (3, -1.5) -- cycle;
\fill[gray!15] (-2, 1) -- (0,0) -- (-1.5,3) -- (-2, 3) -- cycle;
\fill[gray!15] (1,-2) -- (0,0) -- (3,-1.5) -- (3,-2) -- cycle;

\draw[dashed, thick] (-2,-2) -- (3, 3);

\draw[thick,->] (0,0) -- (-0.98, -1.98) node[right] {$\lambda_2$};
\draw[thick,->] (0,0) -- (-1.98, -0.98) node[above] {$\lambda_1$};

\node[right] at (0,0) {O};
\foreach \x in {-2,-1,...,3} {
	\foreach \y in {-2,-1,...,3} {
		\fill (\x, \y) circle (1pt);
	}
}
\end{scope}
\end{tikzpicture}
\caption{Left: A rational mirror-symmetric cone in real-space as an intersection of two half-spaces (in $d>2$ this is a view from the top). Right: CW-structure of $\Omega$.}\label{fig:quarter_space}
\end{figure}


\medskip

The goal of this section is not to provide a K-theoretic classification of corner states and the boundary maps, which has already been done elsewhere \cite{HayashiCMP2018,HayashiLMP2021,HayashiLMP2019,HayashiCMP2022,Dani,DET}. Instead, our goal is to construct interesting Hamiltonians that have second-order boundary states, i.e. ones that are gapped in the bulk and both half-spaces, but have in-gap boundary states localized at the tip of the cone. Note that this is difficult to do analytically; the few known exactly solvable models tend to have a very special algebraic form. A recent K-theoretic approach improves on that point by developing an extended symbol calculus based on Toeplitz factorizations \cite{HayashiCMP2022}. We show here that it is easy and convenient to use adiabatic quantizations to obtain a model for a mirror-symmetric second-order topological insulator on a quarter-space.

An adiabatic symbol $H\in C(\Omega\times \TM^d, M_N(\CM))$ is a self-adjoint matrix function whose restriction to the $1$-skeleton $\Omega_1= \RM_1\sqcup\RM_2\sqcup\{\pm\infty\}$ is invertible at every point. Thus, here we take $\Omega_\infty = \Omega_1$ and recall that $\Omega = \Omega_2$. In the adiabatic model, the relevant boundary map is $$\partial_1\colon K_1(C(\Omega_1)\otimes C(\TM^d))\to K_{0}(C_0(\Omega_2\setminus\Omega_1)\otimes C(\TM^d)).$$
Note that $\Omega_1\simeq \SM^1$ can be oriented by the outward pointing normal of the $2$-cell and by pulling back along an orientation-preserving homeomorphism $\SM^1\to \Omega_1$ we can reduce to the boundary map
$$\partial_{\SM^1}\colon K_1(C({\SM^1})\otimes C(\TM^d))\to K_{0}(C_0(\RM^2)\otimes C(\TM^d))$$
of Proposition~\ref{prop:boundary_sphere}. Together with Theorem~\ref{theorem:adiabatic_chern_numbers} we can conclude:
\begin{proposition}
\label{prop:corner_loop_correspondence_quarter}
The adiabatic boundary map $$\partial: K_i(C(\Omega_1)\otimes C(\TM^d))\to K_{1-i}(C_0(\Omega_2\setminus \Omega_1)\otimes C(\TM^d))$$ satisfies
\begin{equation}
\label{eq:quarterspace_chern_adiabatic}
\langle \partial (\Psi_i(H^{(0)})), [\Ch_{v_1,...,v_m}\#\Ch_{\SM^1}]\rangle \;=\; \langle \Theta_{1-i}(H^{(0)}), [\Ch_{v_1,...,v_m} \# \Ch_{\RM^2}]\rangle.
\end{equation}
If $H=(H^{(t)})_{t\in [0,1]}$ is the quantization of an adiabatic symbol $H^{(0)}\in M_N(C(\Omega)\rtimes \ZM^d))$ as in Proposition~\ref{prop:quantization} which is gapped at $\Omega_1$ then the corner Chern numbers of the quantization are
\begin{equation}
\label{eq:corner_loop_chern_adiabatic}
\langle \Theta_{1-i}(H^{(1)}), [\Ch^{\Lambda}_{v_1,...,v_m}]\rangle\; =\; \langle \Theta_{1-i}(H^{(0)}), [\Ch_{v_1,...,v_m,\lambda_1,\lambda_2} \# \Ch_{\RM^2}]\rangle.
\end{equation}
\end{proposition}

In the following, we assume $d=2$ and that $\lambda_1$ and $\lambda_2$ are related by the diagonal mirror symmetry $(x_1,x_2)\in \RM^2 \mapsto (x_2,x_1)$. As an interesting example, we study a second-order topological insulator protected by a chiral and diagonal mirror symmetry as given in \cite{SchindlerNeupert}: Consider the function $H_B\colon\TM^2 \to M_4(\CM)$ defined as
$$
\begin{aligned}
	H_B(k_1,k_2)\;:=\; & (1+\mu \cos k_1)(\sigma_2 \otimes \sigma_0)+ (1+\mu \cos k_2)(\sigma_2 \otimes \sigma_2) \\
	&-\mu \sin k_1 (\sigma_2 \otimes \sigma_3)+\mu \sin k_2(\sigma_2 \otimes \sigma_1)
\end{aligned}
$$
with the Pauli matrices $\sigma_i$. It has the chiral symmetry $JH_B J=-H_B$ with $J=\sigma_3\otimes \one_2$ and mirror symmetry across the diagonal
\begin{equation}
\label{eq:mirror}
H_B(k_1,k_2)  \;=\; M\, H_B(k_2,k_1)M^*, \qquad M \;=\; \begin{pmatrix}
    0 & 1 & & \\
    1 & 0 & & \\
    &  & 1 & 0 \\
    &  & 0 & -1
\end{pmatrix}.
\end{equation}
For $\mu>1$, this bulk Hamiltonian is gapped and defines a non-trivial class in the $\ZM_2$-equivariant K-theory $K^{\ZM_2}_1(C(\TM^2))$ \cite[Example 4.8]{DET}, where $\ZM_2$ is the mirror symmetry implemented by the operator $M \otimes U$ with $U\colon \ell^2(\ZM^2)\to \ell^2(\ZM^2)$ the unitary induced from the map $\ZM^2\ni (x_1,x_2)\mapsto (x_2,x_1)$. We will now prove the following:
\begin{proposition}
There exists a Hamiltonian $H^{(1)}\in M_4(C(\Omega)\rtimes \ZM^2)$ with the following properties:
\begin{enumerate}
	\item[(i)] $R_{\{\infty\}}(H^{(1)}) = H_B$, the bulk Hamiltonian given above, and $R_{\{-\infty\}}(H^{(1)})$ is a scalar matrix in $M_4(\CM)$.
	\item[(ii)] $\pi_\omega(H^{(1)})$ has a spectral gap for each $\omega\in \Omega_1$.
	\item[(iii)] Each $\pi_\omega(H^{(1)})$ has the chiral symmetry $J \pi_\omega(H^{(1)})J = -\pi_\omega(H^{(1)})$ and the mirror symmetry $$\pi_\omega(H^{(1)}) = (M \otimes U) \pi_\omega(H^{(1)}) (M^*\otimes U^*)$$ 
	\item[(iv)] For each $\omega \in \Omega\setminus \Omega_1$, 
	\begin{equation}
	\label{eq:zero_mode}\langle \pi_\omega(\Theta_{0}(H^{(1)})), {\Ch}^\omega_{\emptyset}\rangle\; =\; \Tr(J\, \mathrm{\Ker}(\pi_\omega(H^{(1)})) \;=\; 1 \mod 2.
	\end{equation}
	Since the kernel projection $\mathrm{\Ker}(\pi_\omega(H^{(1)})$ commutes with $J$, $\pi_\omega(H^{(1)})$ has $0$ as an eigenvalue with odd degeneracy.
\end{enumerate}
\end{proposition}
Let us first elaborate on the significance of this result. Using the formalism of higher-order bulk-boundary correspondence \cite{DET} one can prove that, for any Hamiltonian which satisfies (i) to (iii), the mod-2-parity of \eqref{eq:zero_mode} is completely determined by the $\ZM_2$-equivariant $K$-theory class of $H_B$. While we construct one specific example Hamiltonian using adiabatic quantizations, one can therefore conclude that in fact any Hamiltonian that satisfies (i) to (iii) with any $H_B$ that represents the same class in equivariant $K$-theory will also satisfy (iv). This is an example of so-called intrinsic higher-order bulk-boundary correspondence where symmetry-protected bulk invariants do not force the occurrence of edge modes but instead of corner modes (see 
\cite{BenalcazarScience2017, DET,GeierPRB2018,SchindlerSciAdv2018,Trifunovic} for more background on the mathematics and physics of higher-order topological insulators).

\begin{proof}
We construct  a mirror symmetric adiabatic symbol $H\colon \Omega\times \TM^2\to M_4(\CM)$ which is invertible when restricted to $\Omega_1$ and such that $H\rvert_\infty = H_B$. To make the parametrization easier, we identify the $2$-cell $\RM^2\subset \Omega$ now with the square $(0,2)^2$ in such a way that $\RM_1$ and $\RM_2$ are the bottom and left edge as in Figure~\ref{fig:quarter_space}. The Hamiltonian on each segment $\RM_i\simeq (0,2)$ will be specified by a continuous path. Choose first any invertible self-adjoint matrix $S(1)\in M_4(\CM)$ that is both chirally symmetric and $M$-symmetric, i.e. anti-commutes with $J$ and commutes with $M$, and a continuous path $t\in [0,1]\mapsto S(t)$ between $S(0)=\sigma_1\otimes \sigma_0$ and $S(1)$ within the self-adjoint invertible chirally symmetric matrices. Such a path always exists since the space of invertible matrices is connected. On the bottom of the square we choose the symbol $H\colon [0,2]^2\times \TM^2 \mapsto M_4(\CM)$ to satisfy
$$H(t,0, k)\;:=\; \begin{cases}
(1-t)H_B(k) + t (\sigma_1\otimes \sigma_0) & \text{for }t\in [0,1]\\
S(t-1) & \text{for }t\in [1,2)
\end{cases}$$
which interpolates between $H_B$ and the scalar matrix $S(1)$. This path is pointwise chirally symmetric and lies in the invertible operators, as one can check with a bit of effort.  The exact choice of the path $S$ will be irrelevant as no path of scalar matrices can contribute to the Chern number of interest, since its expression involves momentum derivatives. The second component of the 1-skeleton is obtained from the first as follows
$$H(0,t,k_1,k_2) \;:=\; M H(t,0, k_2,k_1) M^*$$
and since $S(1)$ is $M$-symmetric this path also ends at $S(1)$. Therefore we have specified an adiabatic symbol $H\rvert_{\Omega_1}$ which is chirally symmetric and mirror-symmetric. 

Assume now that one has arbitrarily continued the symbol to any continuous function $\Omega$ which is chiral- and mirror-symmetric (e.g. one can use a linear interpolation between $H\rvert_{\Omega_1}$ and $S(1)$).  The precise continuous extension does not matter for our purposes, since the $1$-skeleton determines the relevant Chern number already by Proposition~\ref{prop:corner_loop_correspondence_quarter}. For any adiabatic quantization $(H^{(t)})_{t\in [0,1]}$ of $H$ we then obtain
\begin{align*}
\langle \Theta_{0}(H^{(1)}), {\Ch}^\Lambda_{\emptyset}\rangle&\;=\;\langle \Theta_{1}(H^{(0)}), \Ch_{\lambda_1 , \lambda_2}\#\Ch_{\RM^2}\rangle\\
&\;=\;\langle \Psi_1(H^{(0)}\rvert_{\Omega_1}), \Ch_{\lambda_1 , \lambda_2}\# \Ch_{\SM^1}\rangle \\
& \; =  \mathrm{Det}(\lambda_1,\lambda_2) \langle \Psi_1(H^{(0)}\rvert_{\Omega_1}), \Ch_{\TM^d}\# \Ch_{\SM^1}\rangle.
\end{align*}
The final expression computes the three-dimensional Chern number of the off-diagonal part of $H\rvert_{\Omega_1}$ over $\SM^1\times \TM^2$ and a scaling factor arises if $\lambda_1$, $\lambda_2$ are not orthonormal. In the present case, the easiest way to evaluate this Chern number is numerical evaluation of the integral formula (which can be easily made rigorous with some effort since one can trivially obtain error estimates for the integration of a smooth function on a compact space and the result is known to be an integer), but there are also other ways. In any case, Chern number evaluates to $1$. Assuming that $\lambda_1,\lambda_2$ are oriented like the standard basis, one has $\mathrm{Det}(\lambda_1,\lambda_2)=\mathrm{Vol}(P)$ in the notation of Proposition~\ref{prop:trace_localization}, which means the determinant cancels if we localize to the orbit of some $\omega\in \RM^2$ inside the $2$-cell to finally obtain
$$\langle \pi_\omega(\Theta_{0}(H^{(1)})), {\Ch}^\omega_{\emptyset}\rangle\; =\; \Tr(J \mathrm{\Ker}(\pi_\omega(H^{(1)}))\; =\; \frac{1}{\mathrm{Vol}(P)}\langle \Theta_{0}(H^{(1)}), {\Ch}^\Lambda_{\emptyset}\rangle \;=\; 1.$$
Here $J$ is the chiral symmetry matrix and the first identity used the non-trivial fact that $0$ is an isolated eigenvalue in the spectrum of the Fredholm operator $\pi_\omega(H^{(1)})$, which simplifies the functional calculus in Definition~\ref{def:ktheory} to the kernel projection for the appropriate choice of $f$.
\end{proof}

It is easy to see that the bulk Hamiltonian $H_B$ can at most determine the parity of the number of zero-energy eigenmodes. After all, while keeping the endpoints of the path between $H_B$ and $S(1)$ the same, one could insert a loop which has itself a non-trvial third Chern number. Thus one can construct an adiabatic symbol for whose quantizations the left-hand side of \eqref{eq:zero_mode} is any integer. However, if we enforce the mirror-symmetry then this process must also introduce a loop on the other side of the square, therefore only changing the number at most by an even integer which preserves the parity.


Let us finally discuss how the adiabatic quantizations look like in real-space. One obtains lattice Hamiltonians with short hopping range which for $\omega\in \Omega\setminus \Omega_1$ interpolate between $H_B$ on the inside of the cone to some scalar matrix on the outside of the cone. This sort of interface model is sometimes inconvenient and one would prefer a model which lives exactly on a quarter-space and has topological zero-modes which are localized to the corner. That can be achieved easily by slightly tuning the adiabatic symbol constructed above to be precisely equal to a constant matrix $S$ outside the third quadrant of the configuration space and using the explicit quantization from the end of Section~\ref{sec-adiabatic}. The resulting lattice Hamiltonians are then of the form
$$\pi_\omega(H^{(1)})\; =\; P^\omega_Q \pi_\omega(H^{(1)}) P^\omega_Q + S (1-P^\omega_Q),$$
where $P^\omega_Q \in \Bb(\ell^2(\ZM^2))$ is the projection to a discrete quarter-space. One can then simply truncate to the invariant subspace $P^\omega_Q$ without affecting number of zero-modes. Similarly, the quantizations from the $1$-cells $\RM_1$, $\RM_2$ result in gapped lattice Hamiltonians that can be truncated to gapped half-space Hamiltonians.
\subsection{A geometry with multiple corners}\label{Sec:MultipleCorners}
In the previous section, we considered an adiabatic model for a quarter-space with a single corner. As emphasized in the recent paper \cite{DET}, it is sometimes advantageous to construct $C^*$-algebras for infinite-volume limits of polyhedra by gluing together $C^*$-algebras for observations on different asymptotic limits of the crystal. The motivation is that many symmetries in crystals, such as inversion symmetry, can leave invariant a finite polyhedron, but never a semi-infinite space such as a half-space or a quarter-space. After gluing together the $C^*$-algebras correctly one can, however, implement the symmetry as an automorphism which exchanges the roles of different corners and thereby unlocks the applicability of equivariant K-theory. In this section, we consider an adiabatic model for one of the simplest examples of this construction, called the infinite square in \cite{DET}, and use it to study hinge modes of a three-dimensional second-order topological insulator.

The idea is to take a crystal of the shape $[-L,L]^{d-2}\times \RM^{d-2}$ which is infinite in all but two directions and consider all geometric features that one would find in the infinite volume limit $L\to \infty$. Such a configuration space needs to contain four $1$-cells corresponding to the $(d-1)$-dimensional faces and four $1$-cells corresponding to the $(d-2)$-dimensional faces.

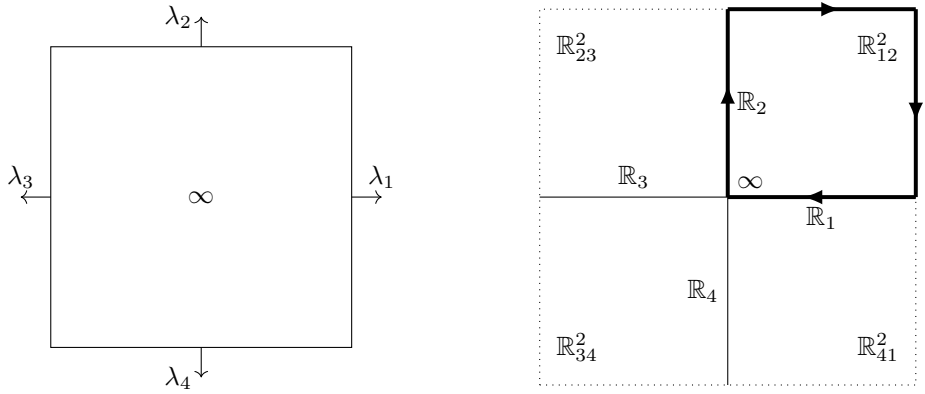
\begin{figure}
	\centering
	\begin{tikzpicture}[arrowmid/.style={postaction={decorate},
			decoration={markings,
				mark=at position .6 with {\arrow{latex}}}},]	
		\begin{scope}[shift={(-2.5,0.5)},scale=0.8]
			\coordinate (A) at (0,0);
			\coordinate (B) at (5,0);
			\coordinate (C) at (5,5);
			\coordinate (D) at (0,5);
			
			\draw (A) -- (B) -- (C) -- (D) -- cycle;
			
			\coordinate (MAB) at ($(A)!0.5!(B)$);
			\coordinate (MBC) at ($(B)!0.5!(C)$);
			\coordinate (MCD) at ($(C)!0.5!(D)$);
			\coordinate (MDA) at ($(D)!0.5!(A)$);
			
			\coordinate (MAB_shifted) at ($(MAB) + (0,-0.5)$);
			\coordinate (MBC_shifted) at ($(MBC) + (0.5,0)$);
			\coordinate (MCD_shifted) at ($(MCD) + (0,0.5)$);
			\coordinate (MDA_shifted) at ($(MDA) + (-0.5,0)$);
			
			\draw[->] (MAB) -- (MAB_shifted);
			\draw[->] (MBC) -- (MBC_shifted);
			\draw[->] (MCD) -- (MCD_shifted);
			\draw[->] (MDA) -- (MDA_shifted);
			
			\coordinate (Center) at ($(A)!0.5!(C)$);
			
			
			\node at (2.5,2.5) {$\infty$};
			
			\node[left] at (MAB_shifted) {$\lambda_4$};
			\node[above] at (MBC_shifted) {$\lambda_1$};
			\node[left] at (MCD_shifted) {$\lambda_2$};
			\node[above] at (MDA_shifted) {$\lambda_3$};
			
			
		\end{scope}
		
		\begin{scope}[shift={(4,0)},scale=1]
			\coordinate (A) at (0,0);
			\coordinate (B) at (5,0);
			\coordinate (C) at (5,5);
			\coordinate (D) at (0,5);
			
			\draw[dotted] (A) -- (B) -- (C) -- (D) -- cycle;
			
			\coordinate (MAB) at ($(A)!0.5!(B)$);
			\coordinate (MBC) at ($(B)!0.5!(C)$);
			\coordinate (MCD) at ($(C)!0.5!(D)$);
			\coordinate (MDA) at ($(D)!0.5!(A)$);
			
			\coordinate (Center) at ($(A)!0.5!(C)$);
			
			\draw (Center) -- (MAB) node[midway,  left] {$\RM_4$};
			\draw (Center) -- (MBC) node[midway, below ] {$\RM_1$};
			\draw (Center) -- (MCD) node[midway,  right] {$\RM_2$};
			\draw (Center) -- (MDA) node[midway, above] {$\RM_3$};
			
			\node at (2.8,2.7) {$\infty$};
			
			\node at (0.5,0.5) {$\RM^2_{34}$};
			\node at (4.5,0.5) {$\RM^2_{41}$};
			\node at (4.5,4.5) {$\RM^2_{12}$};
			\node at (0.5,4.5) {$\RM^2_{23}$};
			
			\draw[arrowmid, line width=1.5pt] (Center) -- (MCD);
			\draw[arrowmid, line width=1.5pt] (MCD) -- (C);
			\draw[arrowmid, line width=1.5pt] (C) -- (MBC);
			\draw[arrowmid, line width=1.5pt] (MBC) -- (Center);
		\end{scope}
	\end{tikzpicture}
	\caption{Visualization of the flat CW-complex $\Omega$ (right) corresponding to infinite-volume limits of a square slab (top view on the left). The exterior square (dotted) is identified with the infinite point of the one-point compactification of $\Omega$ and the loop marks the boundary of $\RM^2_{12}$.}\label{fig:square_adiabatic}
\end{figure}

We therefore use four $1$-cells $(\RM_a)_{a\in \ZM_4}$ with normal vectors $\lambda_a$ pointing in the four different cardinal directions to represent the four possible half-spaces. We also have four $2$-cells, $(\RM^2_{a,a+1})_{a\in \ZM_4}$, which models the corners at the intersections of the respective half-spaces.  The configuration space is then 
\begin{equation}
	\label{eq:config_square}
	\Omega\;=\; \Omega_2 \;=\; \bigsqcup_{a\in \ZM_4} \RM^2_{a,a+1} \bigsqcup_{i\in \ZM_4} \RM_{a} \sqcup \{ +\infty\} \sqcup \{- \infty\}.
\end{equation}
Its topology shall be identified with $\Omega \simeq(-1,1)^2\sqcup\{*\}$ the one-point compactification of $(-1,1)^2$, in such a way that each $\RM^2_{a,a+1}$ becomes a quadrant, the $\RM_i$ are boundary lines of the quadrants, the point $+\infty$ becomes the center of the square and $-\infty$  the outer border $\{*\}$,  see Figure~\ref{fig:square_adiabatic}. This configuration space is a natural push-out of four configuration spaces for different quarter-spaces as they were considered in Section~\ref{sec:quarter}. The $1$-skeleton is
$$\Omega_1\; = \;\bigsqcup_{a\in \ZM_4}  \RM_{a} \sqcup  \{+ \infty\} \sqcup \{ -\infty\}$$
and we consider $\Omega_\infty=\Omega_1$ as the asymptotic part. We are therefore looking for models that are gapped when considered on any of the four half-spaces, but possibly gapless when considered on a quarter-space.  

Constructing the configuration space in this manner allows for the implementation of crystalline symmetries. For example, if $\lambda_1=-\lambda_3$ and $\lambda_2=-\lambda_4$, then the inversion symmetry, which acts as $x\in \RM^d \mapsto -x$, is implemented as a $\ZM_2$-action on $\Omega$ mapping $\RM^2_{a,a+1}$ bijectively to $\RM^2_{a+2,a+3}$ and $\RM_{a}$ to $\RM_{a+2}$. 

For boundary invariants we have to distinguish between the lattice models on the four sectors $\omega \in \RM^2_{a,a+1}$ which are quarter-spaces in different directions. In each case, the adiabatic boundary map relates Chern numbers on the boundary $\partial \RM^2_{a,a+1}\simeq \SM^1$ to those on the interior:

\begin{proposition}
	\label{prop:corner_loop_correspondence}
	
	The adiabatic boundary map $$\partial\colon K_i(C(\Omega_1)\otimes C(\TM^d))\to K_{1-i}(C_0(\Omega\setminus \Omega_1)\otimes C(\TM^d))$$ is given by the direct sum of four boundary maps of the form
	$$\partial_{a}\colon K_i(C(\partial \RM^2_{a,a+1})\otimes C(\TM^d))\to K_{1-i}(C_0(\RM^2_{a,a+1})\otimes C(\TM^d)).$$ 
	The latter is determined by the duality 
	\begin{equation}
			\label{eq:corner_chern_adiabatic}
			\begin{aligned}
		& \langle \partial_{a} (\Psi_i(H^{(0)})\rvert_{\partial \RM^2_{a,a+1}}), [\Ch_{v_1,...,v_m}\#\Ch_{\SM^1}]\rangle \; \\
		& \qquad \qquad =\; \langle \Theta_{1-i}(H^{(0)})\rvert_{\RM^2_{a,a+1}}, [\Ch_{v_1,...,v_m} \# \Ch_{\RM^2_{a,a+1}}]\rangle.
		\end{aligned}
	\end{equation}
	For the defect Chern numbers of the adiabatic quantization of a symbol $H^{(0)}$ which is gapped at $\Omega_1$ one therefore has 
	\begin{equation}
		\label{eq:loop_chern_adiabatic}
		\begin{aligned}
		& \langle \Theta_{1-i}(H^{(1)})\rvert_{\RM^2_{a,a+1}}, [\Ch^{\Lambda_{a,a+1}}_{v_1,...,v_m}]\rangle\; \\
		& \qquad \qquad =\; \langle \Theta_{1-i}(H^{(0)})\rvert_{\RM^2_{a,a+1}}, [\Ch_{v_1,...,v_m,\lambda_a,\lambda_b} \# \Ch_{\RM^2_{a,a+1}}]\rangle.
		\end{aligned}
	\end{equation}
\end{proposition}

We will now give an example which motivates the consideration of this configuration space. Assume $d=3$ and that $\lambda_1,\lambda_2$ are the standard basis vectors $e_1,e_2$ and that $\lambda_3$, $\lambda_4$ are $-e_1$ and $-e_2$ respectively. This makes the infinite square $\Omega$ invariant under all point symmetries of an actual square in a way that is compatible with the $\RM^3$-action, i.e. they will extend to automorphisms of the crossed product algebras as well.

We want to quantize the symbol $H_B\in C(\TM^3)\otimes M_2(\CM)$ given by 
\begin{equation}
	H_B(k)\; =\; \sum_{i=1}^3 \Gamma_i \sin(k_i)+ \Gamma_0 \otimes \Big(2 +\sum_{i=1}^3 \cos(k_i)\Big ),
\end{equation}
where $\Gamma_0=\one_2\otimes \sigma_3$, $\Gamma_1= \sigma_3\otimes \sigma_1$, $\Gamma_2=\one_2\otimes\sigma_2$, $\Gamma_3=\sigma_2\otimes \sigma_1$ and $\sigma_i$ stand for the Pauli matrices. It is gapped and satisfies the inversion symmetry
$$H_B(k)\; =\; \Gamma_0 H_B(-k) \Gamma_0.$$
This model is a slight modification of the model in \cite{DET} which in turn was adapted from \cite{HughesPRB2011}. Its main feature is that its Fermi projection represents a non-trivial class in the equivariant K-theory $K^{\ZM_2}_0(C(\TM^3))$ protected by the $\ZM_2$-inversion symmetry, but as a class in $K_0(C(\TM^3))$ all its Chern numbers vanish. It represents a second-order topological insulator protected by inversion symmetry.

\begin{proposition}
	There exists a Hamiltonian $H^{(1)}\in M_4(C(\Omega)\rtimes \ZM^3)$ with the following properties:
	\begin{enumerate}
		\item[(i)] $R_{\{\infty\}}(H^{(1)}) = H_B$,  the bulk Hamiltonian given above, and $R_{\{-\infty\}}(H^{(1)})$ is a scalar matrix in $M_4(\CM)$.
		\item[(ii)] $\pi_\omega(H^{(1)})$ has a spectral gap for each $\omega\in \Omega_1$.
		\item[(iii)] One has the inversion symmetry
		$$ \pi_{\sigma(\omega)}(H^{(1)})\;=\;(\Gamma_0 \otimes I)\pi_\omega(H^{(1)}) (\Gamma_0\otimes I^*)$$
		where $\sigma:\Omega\to \Omega$ is an involutive homeomorphism exchanging $\RM^2_{a,a+1}$ and $\RM^2_{a+2,a+3}$ and $I\colon\ell^2(\ZM^3)\to \ell^2(\ZM^3)$ implements the inversion symmetry $\ZM^3\ni q\mapsto -q$.
		\item[(iv)] For any $\omega_1 \in \RM^2_{12}$ and $\omega_2 \in \RM^2_{23}$ one has
		\begin{equation}
			\label{eq:hinge_modes}\langle \pi_{\omega_1}(\Theta_{0}(H^{(1)})), {\Ch}^{\omega_1}_{e_3}\rangle\;  + \langle \pi_{\omega_2}(\Theta_{0}(H^{(1)})), {\Ch}^{\omega_2}_{e_3}\rangle \;=\; 1 \mod 2.
		\end{equation}
	\end{enumerate}
\end{proposition}
In real-space, each $\pi_\omega(H^{(1)})$ for $\omega\in \RM^2_{a,a+1}$ interpolates between $H_B$ on the interior of the quarter-space spanned by $\lambda_a$, $\lambda_{a+1}$ and a fixed scalar matrix on the outside. When (iii) holds then the inversion symmetry is preserved globally in the sense that inversion relates Hamiltonians on any two opposing quarter- respectively half-spaces. The statement (iv) then means in real-space that for at least two opposing ones of the four possible quarter-spaces in different directions the lattice Hamiltonian $\pi_\omega(H^{(1)})$ must exhibit non-trivial quantized hinge currents flowing in the direction $e_3$ orthogonal to the quarter-space.

Importantly, if the halfspace models are gapped as in (ii) then the parity in (iv) is by the mechanism of higher-order bulk-boundary correspondence \cite{DET} determined entirely by the $\ZM_2$-equivariant K-theory class of the bulk Hamiltonian $H_B$, thus it is in fact independent of the specific example. This K-theoretic obstruction which leads to a higher-order bulk-boundary correspondence is only possible to study by considering an extended configuration space where the crystalline symmetry can be implemented and therefore motivated the construction. Note that in real-space the lattice Hamiltonians from the four $2$-cells are only loosely connected by the continuity condition and symmetry, however, when one considers large finite-volume approximations on a large slab $[-L,L]^3\cap \ZM^3$ with periodic boundary conditions in the $e_3$-direction and open boundary conditions in $e_1$ and $e_2$ all four quarter-space Hamiltonians effectively become combined into a single Hamiltonian whose spectrum and boundary modes are a mixture of those of the individual quarter-space Hamiltonians.

\begin{proof}
As before, we construct an adiabatic symbol which is invertible at $\Omega_1$. In $\Omega_1$ there are now four path segments and we choose the Hamiltonian on each of them in the same form by identifying each $1$-cell $\RM_{a}\sim (0,2)$ with an open interval, with $0$ representing the point $\infty \in \Omega_1$ and $2$ the point $-\infty \in \Omega_1$. On the interval $[0,1]$ the interpolating Hamiltonian shall be the path
\begin{equation}
\label{eq:path}
(t,k)\in [0,1]\times \TM^3 \mapsto H(t,k)\;:=\; H_B(k) -  t \sum_{i=1}^3 \Gamma_0 \cos(k_i) + 2 \mu t(1-t) M_a
\end{equation}
i.e. we homotopically remove the cosine terms but add a perturbation $M_a$ depending on the line segment $a\in \ZM_4$. Here we use $M_a=s_a (\sigma_1\otimes \sigma_1)$ with $s_a\in \{-1,1\}$ which keeps the gap open for any $\mu>0$. This is rather easily checked after noting that the entire expression squares to a diagonal matrix. On the interval $[1,2]$ we then deform the endpoint of the first path to a scalar matrix:
$$H(k,t)\;=\; (2-t) H(k,1) + (t-1)\Gamma_0.$$
Those four paths define a single consistent gapped adiabatic symbol $H\rvert_{\Omega_1}$ for any choice of signs $s_a$.

It is a simple exercise to compute the four loop Chern numbers numerically with high accuracy which gives
\begin{equation}
	\label{eq:corner_chern}
	\langle \Psi_0(H\rvert_{\partial \RM^2_{a,a+1}}), [\Ch_{e_3, \lambda_i,\lambda_{i+1}}\# \Ch_{\partial \RM^2_{a,a+1}}]\rangle \;=\;
	\frac{1}{2}(s_a-s_{a+1})
\end{equation}
since any of the path segments ends up contributing $\pm s_a \frac{1}{2}$ to the integral formula depending on the orientation. Enforcing the inversion symmetry one has $s_a=-s_{a+2}$ and thus the right-hand side of \eqref{eq:hinge_modes} is equal to $s_a$ which is an odd integer.

Extending the symbol symmetrically from $\Omega_1$ to $\Omega$ quantizing completes the proof.
\end{proof}

Let us also sketch a simple way to derive the four-dimensional Chern numbers in a more topological way. Without the symmetry-breaking mass term, i.e. for $\mu=0$, the spectral gap closes along the path \eqref{eq:path} at exactly one point in the four-dimensional phase space, namely at $(t,k_1,k_2,k_3)=(\frac{1}{3},0,0,0)$ where one finds a four-dimensional band-crossing point with linear dispersion. Expanding around this point the Hamiltonian takes to linear order the form
$$H(\frac{1}{3}+\tilde{t},k) \approx k_1 \Gamma_1 + k_2 \Gamma_2 + k_3 \Gamma_3 + 3 \Gamma_0 \tilde{t}+ s_a\frac{4}{9} \Gamma_5 \mu$$
with the matrices $\Gamma_i$ forming a four-dimensional irreducible representation of the Clifford algebra with five generators. For small $\mu$ a large contribution to the Chern number integral comes from a small ball around this point. The limit of this contribution can be calculated analytically from the above  linear expansion of the Hamiltonian and it converges to $\frac{1}{2}s_a$ for $\mu\downarrow 0$ while the contribution of the remainder of $[0,1]\times \TM^3$ without the small ball is continuous at $\mu=0$ and thus in particular does not depend on $s_a$. When one opens the gap using the mass term, the different sign choices $s_a=-1$ and $s_a=1$ therefore lead to a difference of $1$ in the value of the Chern number integral performed over the path. This observation allows us to reduce the computation of \eqref{eq:corner_chern} to the case $(s_1,s_2,s_3,s_4)=(1,1,1,1)$ for which the Chern numbers of all four loops are trivially zero due to cancellation between the paths for $a$ and $a+1$ which contribute with opposite orientations. For similar examples how one can derive Chern numbers of topological insulators from analyzing how band-crossings split under addition of a mass term let us point towards \cite{Bel95}\cite[2.3.3]{PSbook} and especially \cite[Section VI]{SSt2} which contains a more detailed analytical version of the above argument in an extremely similar setting.

The same strategy can be used to prove the existence of non-trivial second-order boundary states for different symmetry classes. Particularly interesting are symmetries that mix rotations with time-reversal such as $C_2T$- and $C_4T$-symmetry, which are time-reversal composed with rotation by $180$ respectively $90$ degrees. Both of those operations map the infinite square to itself and there are equivariant bulk K-theory classes that are topologically forced to have non-trivial hinge currents. In \cite{DET} one can find example Hamiltonians and one can compute hinge Chern numbers for quarter-space realizations of them as above, in fact, one can use formally the same Hamiltonian as above and just needs to substitute different matrices $\Gamma_0,\Gamma_1,\Gamma_2,\Gamma_3$ and $M_a$. The adiabatic construction also gives some hints as to the topological origin of bulk-hinge correspondence: The equivariant K-theory class at the point $\infty$ must result in some topological obstruction which makes it impossible to extend it to a gapped symmetric symbol on $\Omega_1\times \TM^3$ which does not have non-vanishing Chern numbers along the outer loops. Using the adiabatic picture it is thus much easier to connect to obstructions from equivariant topology than in the full non-commutative operator-algebraic version of the K-theoretic boundary maps, however, this is still outside the scope of the present paper and we will need to leave this to future investigations.

\appendix

\section{K-theory and cyclic cohomology}
\label{app:ktheory}

Let $A$ be a local $C^*$-algebra, i.e. a $*$-algebra which is dense in a $C^*$-algebra $\Aa$ and closed under its holomorphic functional calculus \cite{Bla}. If $A$ is unital then the group $K_0(A)$ is the Grothendieck group generated by homotopy equivalence classes of matrix-valued projections in $\lim_{N\to \infty} M_N(A)$ with the direct sum. Similarly, $K_1$ is the Grothendieck group generated by homotopy equivalence classes of matrix-valued unitaries in $\lim_{N\to \infty} M_N(A)$. For non-unital $A$ one defines $K_i(A)=\Ker(K_i(A^\sim)\to K_i(\CM))$ with $A^\sim$ the minimal unitization. Higher K-groups can be defined by suspensions, but since they satisfy Bott periodicity, we will always set $K_i(A)=K_{i\,\mathrm{mod}\,2}(A)$ as a definition.

It is often more convenient to describe K-theoretic invariants in terms of numerical invariants and one of the prime methods for that is the pairing with cyclic cohomology \cite{Connes94}. A cyclic $m$-cocycle over an algebra $A$ is an $(m+1)$-linear functional $\varphi:A^{m+1}\to \CM$ which is cyclic
$$
\varphi(a_0,\ldots ,a_m)\;=\;(-1)^m \varphi(a_1,\ldots ,a_{m},a_0)
$$
and closed $b\varphi=0$ w.r.t. the Hochschild boundary operator $b$ defined by
\begin{align*}
(b\varphi)(a_0,\ldots ,a_{m+1})\;=\; & \sum_{j=0}^n (-1)^j \varphi(a_0, \ldots , a_j a_{j+1},\ldots , a_{m+1})
\\
& \;+\; (-1)^{m+1} \varphi(a_{m+1} a_0, a_1, \ldots , a_{m}).
\end{align*}
The cohomology classes of cyclic $m$-cocycles are denoted $HC^m(A)$ and there is a numerical pairing $K_m(A)\times HC^m(A)\to \CM$ given by
\begin{equation}
\langle [x]_m, [\varphi]_m\rangle \;=\; \begin{cases}
	\frac{1}{(m/2)!(-2\pi \imath)^{m/2}}\,\,\varphi(x,\dots,x) & m\text{ even}\\
	\frac{\imath^{(m+1)/2}}{m!!(-2)^m\pi^{(m+1)/2}}\,\,\varphi(x^{-1},x,x^{-1},\dots,x,x^{-1}) & m \text{ odd}
\end{cases}
\end{equation}
with any representative $x\in M_N(A^\sim)$ which is a projection in the even case and a unitary otherwise.

An important example is the Chern cocycle on a locally compact oriented $m$-dimensional manifold $X$ given by 
$$(f_0,\dots,f_m)\in C_c^1(X)\; \mapsto\; \Ch_X(f_0,\dots,f_m) \;=\; \int_X f_0 df_1 \dots df_m.$$

The normalization constants of the pairing are chosen such that each of the non-trivial $K$-groups $$K_i(C_0(\RM^n)) \;=\; \begin{cases}
    \ZM & \text{ if }n=i\mod 2\\
    0 & \text{ if }n\neq i\mod 2\
\end{cases}$$
is generated by a single projection respectively unitary whose numerical pairing with $\Ch_{\RM^n}$ takes the value $1$ (sometimes called the Bott or dual-Dirac element).

The following is well-known, e.g. \cite[Section 3.3]{Connes94}, but we could not locate proof in the literature so we use the opportunity to give a few details.
\begin{proposition}
\label{prop:cup_product}
Let $\varphi$, $\eta$ be $n$ and $m$ cocycles over local $C^*$-algebras $A$ and $B$ respectively. Then
$$\langle [e]_{n} \otimes [f]_{m}, \varphi\# \eta\rangle \;=\; \langle [e]_{n}, \varphi\rangle \, \langle [f]_{m}, \varphi\rangle.$$
\end{proposition}
\begin{proof}
    A degree $n$ cyclic cocycle over $A$ can equivalently be considered an $n$-dimensional cycle, i.e. a closed graded trace on the degree $n$ elements of the universal differential graded algebra over $A$ \cite[Section 3.1]{Connes94}. When we identify $A$ and $B$ with commuting subalgebras of $A\otimes B$ as $a=a\otimes 1$ and $b=1\otimes b$, the cup product $\varphi\#\eta$ is the graded trace on the universal differential graded algebra over $A\otimes B$ such that
$$(\varphi\#\eta)( (a_0 da_1...da_m) (b_0 db_1...db_n)) \;=\;\varphi(a_0 da_1...da_m) \eta(b_0 db_1...db_n).$$

One may assume that $A$ and $B$ are unital and one can reduce to the case that $n$ and $m$ are both even since the external product in K-theory is compatible with suspensions \cite[Appendix 2]{Connes81} and there are natural degree-raising suspension maps in cyclic cohomology which among other things commute with cup products. Therefore $e$,$f$ are projections and $[e]_0 \otimes [f]_0 = [e\otimes f]_0$.

To compute the pairing with this class one needs to evaluate the trace of the product
\begin{equation}
\label{eq:powerexpand}
ef (d (ef))^{n+m} \;=\; ef( e df + f de)^{n+m}\;=\; ef (A+B)^{n+m}.
\end{equation}
If $k$ is an odd power then the fact that $e$,$f$ are idempotents and $d$ a graded differential imply
$$e (de)^k e\;=\;0\; = \;f (df)^k f$$
and if $k$ is even then
$$e (de)^k e \;=\; e (de)^k, \quad f(df)^k f = f (df)^k.$$
Due to the first identity the only terms that survive when expanding the product in \eqref{eq:powerexpand} are those that are products of factors $AA$ and $BB$ and since the trace vanishes unless there are exactly $n/2$ of the former and $m/2$ of the latter there are precisely $${\left(\frac{n}{2}+\frac{m}{2}\right)!}/({\frac{n}{2}!})$$
terms with non-vanishing trace. Each of those is equal to $e (de)^n f(df)^m$ since $e$, $f$, $de$ and $df$ graded commute. Hence we find
$$(\varphi\#\eta)(ef (d (ef))^{n+m})\;=\;\frac{\left(\frac{n}{2}+\frac{m}{2}\right)!}{\frac{n}{2}!}\varphi(e (de)^n) \eta(f (df)^m).$$
If $c_n$ is the normalization constant for the pairing with cyclic cohomology then this pre-factor is precisely equal to $\frac{c_n c_m}{c_{n+m}}$.
\end{proof}

This justifies that the K-groups $K_n(C_0(\RM^n))$ are labeled by $\Ch_{\RM^n}$ since its generator is the $n$-fold exterior power of the generator of $K_1(C_0(\RM))$. Let us next revisit the torus:

\medskip 

{\noindent \bf Proof (of Proposition~\ref{prop:Chern_torus}).}
On $C^\infty(\SM^1)$ there are two non-trivial cocycles: a $0$-cycle $\int$ which is the normalized Lebesgue integral and a $1$-cocycle $\Ch_{\SM^1}$ representing the normalized Chern cocycle with
$$\langle K_0(C(\SM^1)), \int \rangle\; =\; \ZM, \qquad \langle K_1(C(\SM^1)), \Ch_{\SM^1}\rangle \;=\; \ZM$$
being one-to-one correspondences. Since $C(\TM^d)=C(\SM^1)^{\otimes d}$ the K-groups of $K_i(C(\TM^d)$ are by the Künneth formula generated by exterior powers of the generators of $K_*(C(\SM^1))$  which can uniquely be labeled by the Chern cocycles $\Ch_J$, which are the cup products of appropriate copies of $\int$ and $\Ch_{\SM^1}$.
\hfill $\Box$

Let us also make explicit the boundary map between the boundary of a sphere and its interior needed in the main text:
\begin{proposition}
\label{prop:sphere_boundary_map}
The boundary map $\partial: K_{n}(C(\SM^n))\to K_{n+1}(C_0(\RM^{n+1}))$ of the exact sequence
$$0 \to C_0(\RM^{n+1}) \to C(\DM^{n+1}) \to C(\SM^n) \to 0$$
is determined uniquely by the duality of Chern cocycles
$$\langle [x]_{n}, \Ch_{\SM^n} \rangle\; =\; \langle \partial([x]_n), \Ch_{\RM^{n+1}}\rangle.$$
\end{proposition}
\begin{proof}
        To fix the relative sign we choose the orientations of $\SM^n$ and $\RM^{n+1}$ such that we can think of $\SM^n$ as the boundary of $\RM^{n+1}$  oriented by the outward pointing normal. In \cite{St5} one can find generators $[Q_n^W]_0$ of $K_n(\SM^n)$ for even $n$ and $[U_n^D]_1$ for odd $n$ such that under the natural inclusion $K_{n+1}(\RM^{n+1})\to K_{n+1}(\SM^{n+1})$ one has $\partial([Q_d^W]_0)=[U^D_{n+1}]_1$ in the even and $\partial([U^D_n]_1)=[Q_{n+1}^W]_0$ in the odd case. Since the inclusion preserves orientations and therefore Chern numbers it remains to compute the pairings of the two sets of generators with the Chern cocycles on $\SM^n$. This is essentially done in \cite[Proposition 3]{CSB} and is not a difficult task since the integrand in the Chern cocycle turns out to be proportional to the volume form of $\SM^n$ and hence one only needs to compute the proportionality constant. Careful evaluation of those with our conventions and the explicit generators from \cite{St5} we obtain $\langle [Q_{2k}^W]_0, \Ch_{\SM^{2k}}\rangle=(-1)^{k-1}$ and $\langle [U_{2k+1}^D]_1, \Ch_{\SM^{2k+1}}\rangle=(-1)^{k-1}$.
\end{proof}

\end{document}